\documentclass[12pt, reqno]{amsart}

\usepackage{amsmath}
\usepackage{amssymb,amsfonts,amscd}
\usepackage{latexsym}
\usepackage{mathrsfs}
\usepackage{amsaddr}

\usepackage{hyperref}

\newcommand\al{\alpha}
\newcommand\be{\beta}
\newcommand\ga{\gamma}
\newcommand\Ga{\Gamma}
\newcommand\de{\delta}

\newcommand\z{\zeta}

\newcommand\Om{\Omega}
\newcommand\la{\lambda}
\newcommand\La{{\Lambda}}
\newcommand{\eps}{\varepsilon}

\newcommand\si{\sigma}
\renewcommand\th{\theta}

\newcommand\R{\mathbb R}
\newcommand\C{\mathbb C}
\newcommand\Z{\mathbb Z}
\newcommand\Y{\mathbb Y}

\newcommand\PP{\mathbb P}
\renewcommand\L{\mathbb L}
\newcommand\MM{\mathbb M}

\renewcommand\P{\mathscr P}

\newcommand\const{\operatorname{const}}

\newcommand\Sym{\operatorname{Sym}}
\newcommand\RTab{\operatorname{RTab}}
\renewcommand\Re{\operatorname{Re}}

\newcommand\abcd{{a,b,c,d}}
\newcommand\alde{{\al,\be,\ga,\de}}

\newcommand\zw{{z,z',w,w'}}

\newcommand\wt{\widetilde}
\newcommand\one{\mathbf1}

\newcommand\LaN{\La^N_{N-1}}
\newcommand\ccdot{\,\cdot\,}
\newcommand\wh{\widehat}

\newtheorem{theorem}{Theorem}[section]
\newtheorem{proposition}[theorem]{Proposition}
\newtheorem{lemma}[theorem]{Lemma}
\newtheorem{corollary}[theorem]{Corollary}

\theoremstyle{definition}
\newtheorem{definition}[theorem]{Definition}
\newtheorem{remark}[theorem]{Remark}

\numberwithin{equation}{section}

\begin{document}

\title[]{Macdonald-level extension of beta ensembles and large-$N$ limit transition}

\author{Grigori Olshanski${}^{1,2,3}$}
\address{${}^1$Institute for Information Transmission Problems of the Russian Academy of Sciences, Moscow, Russia.\\ ${}^2$Skolkovo Institute of Science and Technology, Moscow, Russia.\\
${}^3$ National Research University Higher School of Economics, Moscow, Russia.\\{\rm email: olsh2007@gmail.com}}

\date{}

\begin{abstract}
We introduce and study a family of $(q,t)$-deformed discrete $N$-particle beta ensembles, where $q$ and $t$ are the parameters of Macdonald polynomials. The main result is the existence of a large-$N$ limit transition leading to random point processes with infinitely many particles. 
\end{abstract}


\maketitle

\tableofcontents

\section{Introduction}

We study a family of discrete $N$-particle ensembles depending on the Macdonald parameters $(q,t)$ and some additional parameters. The main result is the existence of a large-$N$ limit transition leading to random point processes with infinitely many particles. 

By an \emph{$N$-particle ensemble} we mean a stochastic system of $N$ distinct points, called \emph{particles}, on the real line. Random matrix theory stimulated the extensive study of \emph{beta-ensembles}. Recall that an $N$-particle beta-ensemble is defined by a probability distribution on $N$-tuples $X=(x_1>\dots>x_N)\in\R^N$, of the form
$$
M_N(dX)=\frac1{Z_N}\,V_\be(X) \prod_{i=1}^N W_N(x_i)dx_i, 
$$
where 
\begin{equation}\label{eq1.V_beta}
V_\be(X):= \prod_{1\le i<j\le N} (x_i-x_j)^\be,  
\end{equation}
$\be>0$ is a parameter, and $W_N(x)$ (the weight function) is a nonnegative function on $\R$. The term $V_\be(X)$ is responsible for the repulsive interaction between particles  --- in the sense that the density of $M_N$ goes to zero as particles approach each other.  

For the three values $\be=1,2,4$, the product $V_\be(X)$ has a simple interpretation: it comes from the radial part of the Lebesgue measure on the space of $N\times N$ Hermitian matrices over the fields $\R$, $\C$, and the skew field $\mathbb H$, respectively. These three values of $\be$ are also special because in this case the correlation functions of the $N$-particle ensembles have determinantal (for $\be=2$) or Pfaffian (for $\be=1,4$) structure, which makes it possible to develop powerful exact methods lacking for general $\be>0$. 

Of primary interest in random matrix theory is the asymptotical behavior of beta-ensembles as $N\to\infty$. In particular, one is interested in the study of point processes with infinitely many particles resulting from large-$N$ limit transitions in beta-ensembles.  There exists a huge literature on this subject for the special values $\be=1,2,4$, see e.g. the monographs Mehta \cite{Mehta}, Deift-Gioev \cite{DG}, Forrester \cite{F-2010}, Anderson-Guionnet-Zeitouni \cite{AGZ}. The case of general $\be>0$ is substantially more complicated, a breakthrough in its study occurred after the advent of the Dumitriu--Edelman tridiagonal model \cite{DE}.

Borodin, Gorin, and Guionnet proposed in \cite{BGG-2017} the following definition of \emph{discrete} $N$-particle ensembles: these are given by discrete probability distributions on $N$-tuples of integers $\nu=(\nu_1\ge\dots\ge\nu_N)\in\Z^N$, of the form 
\begin{equation}\label{eq1.Vdiscr1}
M_N(\nu)=\frac1{Z_N} V^{discr}_\tau(\mathbf n) \prod_{i=1}^N W_N(n_i),
\end{equation}
where
\begin{equation}\label{eq1.Vdiscr2}
\tau:=\tfrac12\be>0, \quad n_i:=\nu_i+(N-i)\tau, \quad \mathbf n:=(n_1,\dots,n_N),
\end{equation}
and
\begin{equation}\label{eq1.Vdiscr}
V^{discr}_\tau(\mathbf n):=\prod_{1\le i<j\le N}\frac{\Ga(n_i-n_j+1)\Ga(n_i-n_j+\tau)}{\Ga(n_i-n_j)\Ga(n_i-n_j+1-\tau)}.
\end{equation}
See the end of subsection 1.2  in \cite{BGG-2017} for a justification of the choice of \eqref{eq1.Vdiscr} as a discrete version of \eqref{eq1.V_beta}.  

In the discrete model, the particle positions are the $n_i$'s. If $\tau$ is a positive integer, then all particles range over the lattice $\Z\subset\R$, but for general $\tau>0$, each particle lives on its own lattice, shifted with respect to $\Z$.

The theory of beta-ensembles is intimately related to the Jack symmetric polynomials. In particular, the term  \eqref{eq1.V_beta}, responsible for the two-particle interaction, has the same analytic form as the orthogonality measure of  the Jack polynomials on the $N$-dimensional torus (Macdonald \cite[ch. VI, (10.36)]{Mac-1995}). 

The main novelty of the present paper is that we establish the existence of a large-$N$ limiting point process  in a model belonging to the higher level, that of the Macdonald polynomials. In particular, the two-particle interaction in our model depends on the two Macdonald parameters $(q,t)$, and its form is similar to that of the orthogonality measure for the Macdonald polynomials on the torus (Macdonald \cite[ch. VI, sect.9]{Mac-1995}). 

Our model may be viewed as a Macdonald-level extension of certain beta-ensembles (continuous and discrete) that arise in asymptotic representation theory. Those ensembles are also related to various Selberg-type integrals. See Borodin-Olshanski \cite{BO-2001}, Neretin \cite{Ner-2003}, Olshanski \cite{Ols-2003a}, \cite{Ols-2003b}, Borodin-Olshanski \cite{BO-2005}, Assiotis-Najnudel \cite{AN}.

The techniques of the paper are essentially algebraic and based on the interplay between several families of multivariable symmetric polynomials: Macdonald polynomials, interpolation polynomials, Koornwinder and big $q$-Jacobi polynomials.  

We proceed to a more detailed description of the contents of the paper. 

\subsection{The model}

We fix two Macdonald parameters $q$ and $t$ inside the open unit interval $(0,1)$. Throughout the paper we use the notation $\tau:=\log_q t$. In the introduction, we will assume that $\tau$ is a positive integer, which makes it possible to  simplify the exposition. However, in the body of the paper we drop this constraint and work with arbitrary $q,t\in(0,1)$.     

\subsubsection{Quantization of the real line: the two-sided $q$-lattice $\L$}\label{sect1.1.1}

Let us fix two extra parameters, $\zeta_+>0$ and $\zeta_-<0$, and consider the subset
\begin{equation}
\L:=\L_-\sqcup \L_+\subset\R, \qquad 
\L_\pm:=\zeta_\pm q^\Z,
\end{equation}
where 
$$
q^\Z:=\{q^n: n\in\Z\}.
$$
We call $\L$ the \emph{two-sided $q$-lattice} and we regard it as a $q$-version (`quantization')  of the real line.

By a \emph{configuration} on $\L$ we mean a subset $X\subset\L$. The points of $X$ and those of its complement  $\L\setminus X$ are interpreted as \emph{particles} and \emph{holes}, respectively. A configuration $X$ is said to be \emph{$\tau$-sparse} if any two neighboring particles are separated by at leat $\tau-1$ holes (this constraint disappears for $\tau=1$). 

\subsubsection{The Macdonald-level particle-particle interaction}

Recall that we temporarily assume that $\tau\in\Z_{\ge1}$ (where $\Z_{\ge1}:=\{1,2,3\dots\}$). It is known (Macdonald \cite[ch. VI, sect. 9]{Mac-1995}) that the Macdonald polynomials with the parameters $(q,q^\tau)$ are orthogonal on the torus 
$$
\mathbb T^N:=\{u=(u_1,\dots,u_N): \; |u_1|=\dots=|u_N|=1\}\subset (\C^*)^N
$$ 
with respect to the scalar product 
$$
(f,g):=\int_{\mathbb T^N}f(u)\overline{g(u)}\Delta(u;q,t) m(du),
$$
where 
\begin{equation}\label{eq1.Delta}
\Delta(u;q,t):=\prod_{1\le i\ne j\le N}\prod_{r=0}^{\tau-1}(1-u_j u_i^{-1} q^r)=\prod_{1\le i\ne j\le N}\prod_{r=0}^{\tau-1}|u_i-u_j q^r|, \quad u\in\mathbb T^N, 
\end{equation}
and $m(du)$ denotes the invariant probability measure on $\mathbb T^N$. Note that $\Delta(u;q,t)\ge0$ for $u\in\mathbb T^N$. 

By analogy with \eqref{eq1.Delta} we set, for an arbitrary $N$-tuple $X=(x_1>\dots>x_N)$ of real numbers, 
\begin{equation}\label{eq1.E}
V_{q,t}(X):=\left|\prod_{1\le i\ne j\le N}\prod_{r=0}^{\tau-1}(x_i-x_j q^r)\right|
=(-1)^{N(N-1)\tau/2} \prod_{1\le i\ne j\le N}\prod_{r=0}^{\tau-1}(x_i-x_j q^r).
\end{equation} 
If $\tau=1$, then the right-hand side does not depend on $q$ and turns into the squared Vandermonde. Note that   
$$
\lim_{q\to1} V_{q,q^\tau}(X)=\prod\limits_{1\le i<j\le N}(x_i-x_j)^{2\tau}, 
$$
so $\tau$ plays the role of $\be/2$. 

Assume now that  $x_1,\dots,x_N\in\L$, so that $X$ is an $N$-particle configuration on $\L$. Then $V_{q,t}(X)$ automatically vanishes when $X$ is not $\tau$-sparse, and is strictly positive whenever $X$ is $\tau$-sparse.

\subsubsection{Hypergeometric $N$-particle ensembles on $\L$}

Throughout the paper we use the standard notation
$$
(u;q)_\infty:=\prod_{n=0}^\infty(1-uq^n).
$$

Let $(\alde)$ be a quadruple of nonzero complex parameters (further conditions on them are specified below).  For each $N=1,2,\dots$ we define a function on $\L$ by
\begin{equation}\label{eq1.A}
W_N^\alde(x):=(1-q)|x|\frac{(\al x;q)_\infty(\be x;q)_\infty}{(\ga t^{1-N}x;q)_\infty(\de t^{1-N}x;q)_\infty}, \quad x\in\L.
\end{equation}

We say that the quadruple $(\alde)$ is \emph{admissible} if the following two conditions hold: 

(1) for any $N\in\Z_{\ge1}$ and any $x\in\L$, the denominator on the right-hand side of \eqref{eq1.A} does not vanish and the whole expression is nonnegative;

(2) let $\Om_N$ denote the set of $\tau$-sparse $N$-particle configurations on $\L$; we require that 
\begin{equation}
Z_N^\alde:=\sum_{X\in\Om_N}V_{q,t}(X)\prod_{x\in X} W_N^\alde(x) <\infty, \qquad N=1,2,\dots\,.
\end{equation} 

Given an admissible quadruple $(\alde)$, we may introduce, for each $N\in\Z_{\ge1}$, a probability measure $M_N^\alde$ on $\Om_N$ with the weights
\begin{equation}\label{eq1.ens}
M_N^\alde(X):=(Z_N^\alde)^{-1}\, V_{q,t}(X)\prod_{x\in X} W_N^\alde(x), \qquad X\in\Om_N.
\end{equation}

The measures $M_N^\alde$ and the corresponding $N$-particle ensembles deserve a special name; let us call them \emph{hypergeometric} measures/ensembles.  

We will not give a detailed description of the whole family of admissible quadruples of parameters but rather single out the following subfamilies, which we call  the principal, complementary, and degenerate series. This terminology is adopted by formal analogy with the nomenclature of representation theory.

\smallskip

$\bullet$ The \emph{principal series} is defined by the constraints
\begin{equation*}
\al=\bar\be\in\C\setminus\R, \quad \ga=\bar\de\in\C\setminus\R, \quad \al\be<\ga\de q.
\end{equation*}
Note that in this case $M_N^\alde(X)>0$ for each $X\in\Om_N$.  

$\bullet$ The \emph{complementary series} is defined likewise: we keep the condition $\al\be<\ga\de q$ and we still want both the numerator and the expression \eqref{eq1.A} for the function $W^\alde_N(x)$ to be strictly positive for all $x\in\L$, so that $M_N^\alde(X)>0$ for each $X\in\Om_N$. The difference is that at least one of the pairs $(\al,\be)$, $(\ga,\de)$ is allowed to be real. For instance, instead of the condition $\al=\bar\be\in\C\setminus\R$ one could require 
\begin{equation}\label{eq1.gammadelta}
\text{$\al,\be\in(\zeta_+^{-1} q^{m+1}, \zeta_+^{-1}q^m)$\; or\; $\al,\be\in(\zeta_-^{-1} q^m, \zeta_-^{-1}q^{m+1})$\; for some $m\in\Z$.}
\end{equation}

$\bullet$ The \emph{degenerate series} is defined by 
\begin{equation}\label{eq1.C}
\al\in\zeta_+^{-1} q^\Z, \quad \be\in\zeta_-^{-1} q^\Z,  \quad \ga=\bar\de\in\C\setminus\R. 
\end{equation}
In particular, $\be<0<\al$.  In contrast to the principal/complementary series, we have now  $M_N^\alde(X)=0$ whenever $X$ is not contained in the closed interval $[\be^{-1}q, \al^{-1}q]$. 

\begin{remark}\label{rem1.A}
The following $N$-fold $q$-integral was suggested by Askey \cite{Ask-SIAM} as a $q$-version of Selberg's beta integral \cite{ForrWar}: 
\begin{equation}\label{eq1.Askey}
\int_{b^{-1}q}^{a^{-1}q}\dots \int_{b^{-1}q}^{a^{-1}q} V_{q,t}(X)\prod_{i=1}^N\frac{(ax_i;q)_\infty(bx_i;q)_\infty}{(cx_i;q)_\infty(dx_i;q)_\infty} d_q x_i, \qquad t\in q^{\Z_{\ge1}}, \quad a>0, \; b<0
\end{equation}
(we only changed Askey's notation to make it closer to our notation). Askey \cite[Conjecture 8, p. 948]{Ask-SIAM} conjectured an explicit expression for the value of this integral, later proved by Evans \cite{Evans}.  The integrand in \eqref{eq1.Askey} is nothing else than our degenerate series hypergeometric measure (up to normalization), and Evans' evaluation of the integral \eqref{eq1.Askey} just yields an explicit expression for our normalization constant $Z^{q,t;\alde}_N$. About the case of general parameters $q,t\in(0.1)$ see Remark \ref{rem5.A}. 
\end{remark}

\begin{remark}
In a limit transition, the hypergeometric ensembles \eqref{eq1.ens} degenerate into certain discrete beta ensembles of type \eqref{eq1.Vdiscr1}, considered in the author's paper \cite{Ols-2003b}; for more detail, see section \ref{sect7.2} below. There is also another limit transition, which leads to  certain continuous beta ensembles of type \eqref{eq1.V_beta}, considered in Assiotis--Najnudel \cite{AN}; about this, see section \ref{sect7.3}.
\end{remark}

\subsection{The main result: large-$N$ transition}

First, we need to introduce an appropriate space of configurations containing all $\Om_N$'s. This is done in the following way.

We say that a particle configuration $X\subset \L$ is \emph{bounded} if it is contained in a bounded interval of $\R$. Note that the boundedness condition does not prevent $X$ from being infinite, because particles may accumulate to $0$. 

We denote by $\wt\Om$ the set of all bounded $\tau$-sparse configurations on $\L$. It is equipped with a topology making it a locally compact separable space. The space $\wt\Om$ has a stratification,
$$
\wt\Om=\Om_\infty\sqcup \bigsqcup_{N=0}^\infty\Om_N,
$$
where $\Om_N\subset\wt\Om$ is the set of $N$-particle configurations and $\Om_\infty\subset\wt\Om$ consists of configurations with infinitely many particles. Both $\Om_\infty$ and its complement $\bigsqcup_{N=0}^\infty\Om_N$ are dense in $\wt\Om$.

Let $(\alde)$ be a fixed admissible quadruple of parameters. Since  all sets $\Om_N$ are contained in the space $\wt\Om$, we may regard each measure $M_N^\alde$ as a probability measure on $\wt\Om$. Let $\P(\wt\Om)$ be the space of probability Borel measures on $\wt\Om$ equipped with the weak topology. 

Our main result is Theorem \ref{thm6.B}. We state here its  version for the case of $\tau\in\Z_{\ge1}$. 

\begin{theorem}\label{thm1.A}
As $N\to\infty$, the measures $M_N^\alde$ converge, in the weak topology of $\P(\wt\Om)$, to a probability measure $M_\infty^\alde$, supported by $\Om_\infty$. 
\end{theorem}

In Theorem \ref{thm6.B} we prove a similar claim for arbitrary $q,t\in(0,1)$, without the constraint $t\in q^{\Z_{\ge1}}$.  Then the picture becomes more complicated: each particle now lives on its own $q$-lattice (of the form $\zeta_\pm t^m q^\Z$, where $m$ is a fixed nonnegative integer depending on the particle's label),  the definition of $V_{q,t}(X)$ given in \eqref{eq1.E} is replaced by a more sophisticated one, etc. 

\subsection{Method of proof}

\subsubsection{Projective chains and boundaries}\label{sect1.3.1}
Here we briefly describe the general formalism, within which we prove the existence of the desired large-$N$ limit. 

As above, $\P(\ccdot)$ denotes the set of probability measures on a measurable space. A \emph{Markov kernel} between two measurable spaces, $S$ and $S'$, is a measurable map $L: S\to \P(S')$ (Meyer \cite{Mey-1966}). We represent $L$ by a dashed arrow $S\dashrightarrow S'$ and treat $L$ as a `generalized map' from $S$ to $S'$. Note that it induces a conventional map $\P(S)\to\P(S')$.

In the particular case when $S$ and $S'$ are finite or countable sets, a Markov kernel $L: S\dasharrow S'$ is given by a stochastic matrix $[L(s,s')]$ of the format $S\times S'$.

By a \emph{projective chain} we mean an infinite sequence $S_1, S_2, S_3, \dots$ of finite or countable sets linked by stochastic matrices $L^N_{N-1}: S_N\dasharrow S_{N-1}$: 
\begin{equation}\label{eq1.K}
S_1\stackrel{L^2_1}{\dashleftarrow}S_2\stackrel{L^3_2}{\dashleftarrow}\cdots\stackrel{L^{N-1}_{N-2}}{\dashleftarrow}S_{N-1}\stackrel{L^N_{N-1}}{\dashleftarrow} S_N \stackrel{L^{N+1}_N}{\dashleftarrow}\cdots\,.
\end{equation}
For the matrices $L^N_{N-1}$, we use the term \emph{stochastic links}.  

For any projective chain one can define, in a canonical way, its \emph{boundary} (Winkler \cite[ch. 4]{Win-1985}). It is a Borel space $S_\infty$, which comes with Markov kernels $\La^\infty_N: S_\infty\dashrightarrow S_N$ satisfying the relations $L^\infty_N L^N_{N-1}=L^\infty_{N-1}$ (a composition of Markov kernels is read from left to right). The boundary $S_\infty$ can be interpreted as the inverse limit of \eqref{eq1.K} in the category-theoretical sense: the corresponding category is formed by standard Borel spaces (as objects) and Markov kernels (as morphisms), see Winkler \cite[ch. 4]{Win-1985}. 

By a \emph{coherent system of measures} we mean a sequence $\{M_N\in\P(S_N): N\in\Z_{\ge1}\}$, such that  
\begin{equation}\label{eq1.D}
\sum_{s\in S_N}M_N(s) L^N_{N-1}(s,s') =M_{N-1}(s'), \qquad  N\ge2, \quad s'\in S_{N-1}.
\end{equation}
This relation can be abbreviate as $M_N L^N_{N-1}=M_{N-1}$, where we treat measures as row vectors. We call \eqref{eq1.D} the \emph{coherency relation}. 

The characteristic property of the boundary is that there is a bijective correspondence $\{M_N\}\leftrightarrow M_\infty$ between coherent systems $\{M_N\}$ and probability measures $M_\infty\in\P(S_\infty)$, given by the relations $M_N=M_\infty L^\infty_N$. We call $M_\infty$ the \emph{boundary measure} of $\{M_N\}$. 

This can be rephrased as the identification
$$
\varprojlim \P(S_N)=\P(S_\infty),
$$
where the maps $\P(S_N)\to \P(S_{N-1})$ are defined by the stochastic links.  

For more details about this formalism see Dynkin \cite{Dyn-1971}, \cite{Dyn-1978}, Winkler \cite{Win-1985}, Olshanski \cite{Ols-2003a}, Borodin--Olshanski \cite{BO-2005}, \cite{BO-Book}.

\subsubsection{Preliminary results}\label{sect1.3.2}

We use the following three theorems from the companion paper \cite{Ols-MacdonaldOne}. 

Theorem A constructs a family of stochastic links $\LaN: \Om_N\dasharrow \Om_{N-1}$ depending on $(q,t)$ and $(\zeta_+,\zeta_-)$. They admit an explicit description, but more important is their characterization in terms of their action on Macdonald polynomials, which looks as follows. 

Let $\Sym(N)$ denote the algebra of symmetric polynomials in $N$ variables, realized as a space of functions on $\Om_N$. The stochastic matrix $\LaN$ determines a linear map 
$$
\Sym(N-1) \to \Sym(N), \qquad F\mapsto \LaN F,
$$
and it turns out that 
\begin{equation}\label{eq1.Macd}
\LaN \wt P_{\nu\mid N-1}=\wt P_{\nu\mid N} \qquad \text{for any partition $\nu$ of length at most $N-1$},
\end{equation}
where $\wt P_{\nu\mid N-1}$ denotes the slightly renormalized $(N-1)$-variate symmetric Macdonald polynomial indexed by $\nu$, and $\wt P_{\nu\mid N}$ is the similar polynomial in $N$ variables and with \emph{the same} index $\nu$.   

Theorem B identifies the boundary of the projective chain 
\begin{equation}\label{eq1.F}
\Om_1\stackrel{\La^2_1}{\dashleftarrow}\Om_2\stackrel{\La^3_2}{\dashleftarrow}\cdots\stackrel{\La^{N-1}_{N-2}}{\dashleftarrow}\Om_{N-1}\stackrel{\La^N_{N-1}}{\dashleftarrow} \Om_N \stackrel{L^{N+1}_N}{\dashleftarrow}\cdots\,.
\end{equation}
with the space $\Om_\infty$. 

Finally, Theorem C asserts that  if $\{M_N\}$ is an arbitrary coherent system for the chain \eqref{eq1.F}, then the measures $M_N$ weakly converge to their boundary measure $M_\infty$. 

\subsubsection{Main ideas of the proof}\label{sect1.3.3}

By virtue of the results outlined in subsections \ref{sect1.3.1} and \ref{sect1.3.2}, to prove Theorem \ref{thm6.B} it suffices to show that the hypergeometric measures satisfy the coherency relation with the links $\LaN$:
\begin{equation}\label{eq1.coherency}
M^\alde_N \LaN=M^\alde_{N-1}.
\end{equation}

With all terms written  explicitly, \eqref{eq1.coherency} looks as an intricate hypergeometric identity. However, we do not attempt to check it directly, our argument relies on \eqref{eq1.Macd} and some properties of the measures $M^\alde_N$, which are also linked with Macdonald polynomials. 

The first idea is to treat the measures of the principal (or complementary) series as an analytic continuation of the degenerate series measures. Using this trick, one can obtain the coherency relation for the principal/complementary series  from the one for the degenerate series. This enables us to focus on the case of the degenerate series measures.  

Next, we use the fact that the degenerate series measures serve as the orthogonality measures for a system of multivariate orthogonal polynomials, the big $q$-Jacobi polynomials (Stokman \cite{St}). 

As a consequence, our task reduces to showing that the big $q$-Jacobi orthogonality measures with suitably adjusted parameters are consistent with the stochastic links  $\LaN$. This claim, in turn, is obtained by combining the following two results:

(1) the consistency of the links $\LaN$ with the Macdonald polynomials, expressed by \eqref{eq1.Macd}; 

(2) a stability property of the expansion of the big $q$-Jacobi polynomials on Macdonald polynomials: after a little renormalization of the polynomials, the coefficients of the expansion do not depend on $N$.  

The latter fact is the core of the proof.

\subsection{Comments and open questions} 

1. In the particular case of equal Macdonald parameters, $t=q$, the claim of Theorem \ref{thm1.A} was proved earlier by Gorin and the author in \cite{GO-2016}. The techniques of the present paper are different from those of \cite{GO-2016}. 

2. The limit measures $M_\infty^\alde\in\P(\Om_\infty)$ afforded by Theorem \ref{thm6.B} give rise to a new family of point processes with infinitely many particles.  

3. In the case $t=q$, the $N$-particle hypergeometric ensembles are discrete orthogonal polynomial ensembles. As a consequence, their correlation functions have determinantal structure. This property is inherited by the limit point process defined by $M_\infty^\alde$, and the limit correlation kernel was explicitly computed in \cite{GO-2016}. 

4. Another special case is that of $t=q^2$. Then the $N$-particle ensembles are Pfaffian ensembles. It would be interesting to compute the corresponding Pfaffian correlation kernels and study their large-$N$ limits, as it was done in \cite{GO-2016} for the determinantal case. 

5. Further results about the limit process with $t=q$ were obtained by Cuenca, Gorin and the author in \cite{CGO}. In particular, it was shown there that the measure $M^\alde_\infty$ is \emph{diffuse} (that is, has no atoms). It seems plausible that this property holds for general values of $(q,t)$.

6. In \cite{CGO}, we studied the \emph{tail process} which is responsible for the asymptotics of the `$t=q$' limit process near $0$. This is a point  process  on the disjoint union $\Z\sqcup \Z$, which is \emph{stationary} in the sense that it is invariant under simultaneous shifts along the two copies of the lattice $\Z$. It seems plausible that a similar tail process exists for $t\in q^{\Z_{\ge1}}$, and it would be interesting to construct it. Perhaps this can be done for $t=q^2$ by a suitable extension of the techniques from \cite{CGO}. But for arbitrary $t\in q^{\Z_{\ge1}}$ one needs to find a new approach.  

7. With each limit measure $M^\alde_\infty$ from the degenerate series one can associate an inhomogeneous basis in the algebra of symmetric functions. The elements of this basis, called \emph{big $q$-Jacobi symmetric functions}, are orthogonal with respect to $M^\alde_\infty$ (Theorem \ref{thm7.A}). This is an example of an \emph{infinite-variate} version of symmetric  orthogonal polynomials; other examples are discussed in \cite{CO}. 

8. Borodin and Corwin \cite{BC} introduced and studied \emph{Macdonald processes} --- a family of probability measures on sequences of partitions, directly related to Macdonald symmetric functions with two parameters $(q,t)$. It is unclear whether it is possible to relate our hypergeometric ensembles to (a particular case of) Macdonald processes (cf. \cite{BO-2013}).  

9. A general model of discrete beta ensembles with two Macdonald parameters $(q,t)$ was considered by Borodin, Gorin, and Guionnet in \cite[sect. 4]{BGG-2017}. Their definition \cite[(82)]{BGG-2017} of the particle-particle interaction is similar to \eqref{eq1.Vdiscr}, but with Gamma functions replaced by $q$-Gamma functions. Our expression for the particle-particle interaction (\eqref{eq1.E} and \eqref{eq5.D}) agrees with  \cite[(82)]{BGG-2017} for configurations $X$ on the right semi-axis $\R_{>0}$.  Dimitrov and Knizel \cite{DK} studied  general $(q,t)$-deformed discrete beta ensembles in which the particle-particle interaction corresponds to a higher level of the hypergeometric hierarchy. However, the problem settings and the results of \cite{BGG-2017} and \cite{DK} are very different from what is being done in the present paper.

\subsection{Organization of the paper} 

The goal of sections \ref{sect2} and \ref{sect3} is to prove the key stability property for the expansion of the $N$-variate big $q$-Jacobi polynomials on the Macdonald polynomials (see subsection \ref{sect1.3.3}). 

Section \ref{sect2} collects necessary results about the $A$-type interpolation Macdonald polynomials. Section \ref{sect3} contains the main computation. We introduce the big $q$-Jacobi polynomials as a degeneration of the Koornwinder polynomials. From this we derive the desired expansion in two steps: first, we expand the big $q$-Jacobi polynomials on the interpolation polynomials, and then the latter are expanded on the Macdonald polynomials. The final result, which we need for the proof of the main theorem, is Corollary \ref{cor3.A}. 

In sections \ref{sect4} and \ref{sect5} we introduce the hypergeometric ensembles and prove some necessary estimates. We do it first for $N=1$ (in section \ref{sect4}) and then move to the $N$-dimensional case (in section \ref{sect5}). In the end of section \ref{sect5} we discuss, without proofs, two degenerations of the hypergeometric measures giving a descent from the Macdonald level to the Jack level. 

After this preparation we prove, in section \ref{sect6}, the main result. 

The last section \ref{sect7} contains a few concluding remarks without proofs. Here we briefly discuss the following topics:
 
$\bullet$ a lifting of $N$-variate big $q$-Jacobi polynomials to the algebra of symmetric functions (this extends the construction of  \cite{Ols-FAA-2017}) and 
 
$\bullet$ the degeneration of the hypergeometric $N$-particle ensembles to discrete and continuous beta ensembles.

\section{Multivariable interpolation polynomials}\label{sect2}

\subsection{Basic notation}\label{sect2.1}

We set $\Z_{\ge1}:=\{1,2,3,\dots\}$ and $\Z_{\ge0}:=\{0,1,2,\dots\}$. We denote by $\Y$  the set of partitions, and we identify partitions with their Young diagrams, as in \cite{Mac-1995}. Elements of $\Y$ will be denoted by the lowercase Greek letters $\la,\mu,\nu$. Next, we set:

\smallskip

$\bullet$ $\ell(\la):=\min\{i: \la_{i+1}=0\}$ is the length of the partition $\la$;

$\bullet$ $|\la|:=\sum_i\la_i$ is the size of $\la$ (i.e. the number of boxes in the corresponding Young diagram);

$\bullet$ $n(\la):=\sum_i(i-1)\la_i$;

$\bullet$ $\la\supseteq\mu$ or $\mu\subseteq\la$ means that diagram $\la$ contains diagram $\mu$;

$\bullet$ $\la'$ is the transpose of diagram $\la$;

$\bullet$ $(i,j)\in\la$ means that the box $(i,j)$ with the row number $i$ and the column number $j$ is contained in $\la$;

$\bullet$ $q$ and $t$ are two parameters; unless otherwise stated, we assume $0<q<1$, $0<t<1$;  

$\bullet$ $\Y(N)\subset \Y$ is the subset of partitions whose length is at most $N$;

$\bullet$  $\Sym(N)$ is the algebra of $N$-variable symmetric polynomials over the base field $\R$ or $\C$ (depending on the context); 

$\bullet$ for an arbitrary number $u\in\C$,
\begin{gather}
(u;q)_n:=\prod_{i=1}^n(1-uq^{i-1}), \quad n\in\Z_{\ge0}; \qquad (u;q)_\infty:=\prod_{i=1}^\infty(1-uq^{i-1}); \notag\\
(u;q,t)_\nu:=\prod_{i=1}^{\ell(\nu)}(u t^{1-i};q)_{\nu_i}, \quad \nu\in\Y; \notag
\end{gather}

$\bullet$ $P_{\nu\mid N}(\ccdot;q,t)$ is the $N$-variable symmetric Macdonald polynomial indexed by partition $\nu\in\Y(N)$, $N=1,2,\dots$; 

$\bullet$ for $\la\in\Y(N)$, 
$$
X_N(\la):=(q^{-\la_1},\, q^{-\la_2}t,\, q^{-\la_3}t^2,\,\dots\, q^{-\la_N}t^{N-1})\in\R^N.
$$

\subsection{Symmetric interpolation polynomials}

These polynomials are a far-reaching generalization of the univariate polynomials
$$
(x-1)(x-q^{-1})\dots(x-q^{1-n}), \qquad n=1,2,\dots,
$$
which are a special instance of the Newton interpolation polynomials. The theory of symmetric interpolation polynomials of several variables is due to Knop, Okounkov, and Sahi. Here is a brief overview of a number of their results that we need, presented in a slightly reformulated form.  

For proofs, see Knop \cite{Knop-CMH}, Sahi \cite{Sahi-IMRN}, Okounkov \cite{Ok-MRL}, \cite{Ok-CM}, \cite{Ok-AAM}. Note also the earlier paper Sahi \cite{Sahi-PM} which is used in the proof of the existence of the interpolation polynomials. 

In the special case $t=q$ the Macdonald polynomials turn into the Schur polynomials and all results stated below can be proved directly in an elementary way (see \cite{Ols-FAA-2017}). 

Our notation for the interpolation polynomials is 
$$
I_{\mu\mid N}(\ccdot;q,t)=I_{\mu\mid N}(x_1,\dots,x_N;q,t), \qquad \mu\in\Y(N).
$$
In the notation of Okounkov \cite{Ok-MRL}, \cite{Ok-CM},
\begin{equation}\label{eq2.B}
I_{\mu\mid N}(x_1,\dots,x_N;q,t)=P^*_\mu(x_1, x_2 t^{-1},\dots,x_N t^{1-N}; q^{-1}, t^{-1}), \qquad \mu\in\Y(N).
\end{equation}

The polynomials $I_{\mu\mid N}(\ccdot;q,t)$ are characterized (within normalization) by the following properties: they are symmetric and 
\smallskip

$\bullet$ $I_{\mu\mid N}(\ccdot;q,t)$ has degree $|\mu|$;

$\bullet$ $I_{\mu\mid N}(X_N(\la))=0$ for all $\la\in\Y(N)$ such that $|\la|\le|\mu|$ and $\la\ne\mu$;

$\bullet$ $I_{\mu\mid N}(X_N(\mu))\ne0$.

\smallskip

We call these polynomials the \emph{$N$-variate interpolation polynomials}. Their normalization is indicated below. 
Further properties are the following. 

\smallskip

$\bullet$ \emph{Extra vanishing property}: $I_{\mu\mid N}(X_N(\la);q,t)$ vanishes whenever the diagram $\la$ does not contain the diagram $\mu$.

$\bullet$ \emph{Connection with Macdonald polynomials}: the top homogeneous component of $I_{\mu\mid N}(\ccdot;q,t)$ is proportional to the Macdonald polynomial $P_{\mu\mid N}(\ccdot;q,t)$. 

The normalization is defined so that the top homogeneous component of $I_{\mu\mid N}(\ccdot;q,t)$ is exactly $P_{\mu\mid N}(\ccdot;q,t)$. 

$\bullet$ \emph{Combinatorial formula}:
\begin{equation}\label{eq2.A}
I_{\mu\mid N}(x_1,\dots,x_N;q,t)=\sum_{T\in\RTab(\mu,N)}\psi_T(q,t)\prod_{(i,j)\in\mu}(x_{T(i,j)}-q^{1-j}t^{T(i,j)+i-2}),
\end{equation}
where $\RTab(\mu,N)$ is the set of reverse semistandard tableaux of shape $\mu$ and with values in $\{1,\dots,N\}$ (here `reverse' means that the entries of $T$ weakly decrease along the rows and strongly decrease down the columns) and $\psi_T(q,t)$ is the weight of $T$ entering the combinatorial formula for Macdonald polynomials, which reads in our notation as
$$
P_{\mu\mid N}(x_1,\dots,x_N;q,t)=\sum_{T\in\RTab(\mu,N)}\psi_T(q,t)\prod_{(i,j)\in\mu}x_{T(i,j)}
$$
(cf. \cite[Section VI.7]{Mac-1995}). (In order to translate \eqref{eq2.A} into the language of conventional tableaux it suffices to replace the exponent of $t$ with $N-T(i,j)+i-1$.)

$\bullet$ \emph{Expansion on Macdonald polynomials}: see \eqref{eq2.D} below. 

$\bullet$ \emph{Stability}: If $\mu\in\Y(N)$, then
$$
I_{\mu\mid N}(x_1,\dots,x_N)\big|_{x_N=t^{N-1}}=\begin{cases} I_{\mu\mid N-1}(x_1,\dots,x_{N-1}), & \mu\in\Y(N-1), \\
0, & \text{otherwise}.
\end{cases}
$$
As shown in \cite{Ols-2019}, the stability property makes it possible to define interpolation \emph{symmetric functions} ---  large-$N$ limits of the polynomials $I_{\mu\mid N}(\ccdot;q,t)$ in the algebra of symmetric functions, but in the present paper we do not use this result. 

\subsection{Expansion of interpolation polynomials in the basis of Macdonald polynomials.} 
In formula \eqref{eq2.C} below it is convenient to temporarily use Okounkov's notation (see \eqref{eq2.B}). Then the stability property reads as
$$
P^*_\mu(y_1,\dots,y_{N-1},1;q,t)=P^*_{\mu}(y_1,\dots,y_{N-1};q,t)
$$
with the understanding that $P^*_\mu(y_1,\dots,y_N;q,t)\equiv0$ if $\ell(\mu)>N$. Due to this property, the quantity $P^*_\mu(y_1,y_2,\dots;q,t)$ is well defined for any $\mu\in\Y$ and any infinite sequence $(y_1,y_2,\dots)$ with finitely many terms distinct from $1$.  

For arbitrary $\mu,\,\nu\in\Y$, we set
$$
q^\mu:=(q^{\mu_1},q^{\mu_2},\dots)
$$
and
\begin{equation}\label{eq2.C}
\si(\mu,\nu;q,t):=\frac{(-1)^{|\mu|}t^{2n(\mu)}q^{-n(\mu')}}{(-1)^{|\nu|}t^{2n(\nu)} q^{-n(\nu')}}\,\frac{P^*_\nu(q^\mu;q,t)}{P^*_\nu(q^\nu;q,t)}\, \frac{\prod\limits_{(i,j)\in\nu}(1-q^{\nu_i-j}t^{\nu'_j-i+1})}{\prod\limits_{(i,j)\in\mu}(1-q^{\mu_i-j}t^{\mu'_j-i+1})}.
\end{equation}

Note that 
\begin{equation}\label{eq2.E}
\text{$\si(\mu,\mu;q,t)=1$ and $\si(\mu,\nu;q,t)=0$ unless $\mu\supseteq\nu$}, 
\end{equation}
where the second equality follows from the extra vanishing property for the polynomials $P^*_\nu(\ccdot; q,t)$.  

\begin{proposition}
For $\mu\in\Y(N)$, the following expansion holds
\begin{equation}\label{eq2.D}
\frac{I_{\mu\mid N}(\ccdot;q,t)}{(t^N;q,t)_\mu}=\sum_{\nu\subseteq\mu}\si(\mu,\nu;q,t)\frac{P_{\nu\mid N}(\ccdot;q,t)}{(t^N;q,t)_\nu}.
\end{equation}
Or, equivalently,
\begin{equation}\label{eq2.F}
I_{\mu\mid N}(\ccdot;q,t)=\sum_{\nu\subseteq\mu}\frac{(t^N;q,t)_\mu}{(t^N;q,t)_\nu}\si(\mu,\nu;q,t)P_{\nu\mid N}(\ccdot;q,t).
\end{equation}
\end{proposition}

\begin{proof}
This is a reformulation of formula (1.12) in Okounkov \cite{Ok-MRL}. Indeed, that formula is written as 
$$
\frac{P^*_\mu(y_1,\dots,y_N;q,t)}{P^*_\mu(0^N;q,t)}=\sum_{\nu\subseteq\mu} \frac{P^*_\nu(q^{-\mu};q^{-1},t^{-1})}{P^*_\nu(q^{-\nu};q^{-1},t^{-1})}\,\frac{P_{\nu\mid N}(y_1,y_2 t^{-1},\dots y_N t^{1-N};q,t)}{P^*_\nu(0^N;q,t)},
$$
where $0^N:=(0,\dots,0)$ ($N$ times). We replace here $(q,t)$ by $(q^{-1}, t^{-1})$, then set $y_i=x_i t^{i-1}$, and  use the fact that the Macdonald polynomials are invariant under the transformation $(q,t)\to (q^{-1},t^{-1})$. Taking into account \eqref{eq2.B}, this gives us  
$$
\frac{I_{\mu\mid N}(x_1,\dots,x_N;q,t)}{P^*_\mu(0^N;q^{-1},t^{-1})}=\sum_{\nu\subseteq\mu} \frac{P^*_\nu(q^\mu; q,t)}{P^*_\nu(q^\nu;q,t)}\,\frac{P_{\nu\mid N}(x_1,\dots x_N;q,t)}{P^*_\nu(0^N;q^{-1},t^{-1})}.
$$ 

Next, by virtue of formula (1.9) in \cite{Ok-MRL},
$$
P^*_\mu(0^N; q^{-1},t^{-1})=(-1)^{|\mu|}t^{n(\mu)}q^{-n(\mu')}P_{\mu\mid N}(1,t,\dots,t^{N-1};q,t),
$$
and, by virtue of the principal specialization formula, see  (6.11$'$) in \cite[Chapter VI]{Mac-1995},
$$
P_{\mu\mid N}(1,t,\dots,t^{N-1};q,t)=t^{n(\mu)}\frac{(t^N;q,t)_\mu}{\prod\limits_{(i,j)\in\mu}(1-q^{\mu_i-j}t^{\mu'_j-i+1})},
$$
so that
$$
P^*_\mu(0^N; q^{-1},t^{-1})=\frac{(-1)^{|\mu|}t^{2n(\mu)}q^{-n(\mu')} (t^N;q,t)_\mu}{\prod\limits_{(i,j)\in\mu}(1-q^{\mu_i-j}t^{\mu'_j-i+1})}.
$$

This finally gives \eqref{eq2.D}.   
\end{proof}

\section{Multivariable big $q$-Jacobi polynomials}\label{sect3}

\subsection{$BC$-type symmetric interpolation polynomials $\bar P^{*(N)}_\mu(\ccdot; q,t, s)$}
Here we collect a number of results from Okounkov \cite{Ok-TG-1998} and Rains \cite{Rains}. We follow the notation of \cite{Rains}.

A Laurent polynomial in $N$ variables $z_1,\dots,z_N$ is said to be \emph{$BC_N$-invariant} if it is invariant with respect to the permutations of variables $z_i$ and inversions $z_i\to z^{-1}_i$. Equivalently,  $BC_N$-invariant polynomials may  be treated as symmetric polynomials in the variables $z_1+z^{-1}_1$, \dots, $z_N+z^{-1}_N$. 

Let $q$ and $t$ be the standard Macdonald parameters, and $s$ be an extra parameter. For each $\mu\in\Y(N)$ there exists a $BC_N$-invariant Laurent polynomial $\bar P^{*(N)}_\mu(z_1,\dots,z_N; q,t, s)$ characterized by the following properties.

\smallskip

$\bullet$ \emph{Condition on degree}: $\bar P^{*(N)}_\mu(\ccdot;q,t, s)$ has degree $|\mu|$ with respect to the variables $z_1+z^{-1}_1$, \dots, $z_N+z^{-1}_N$.

$\bullet$ \emph{Vanishing property}: Given $\la\in\Y(N)$, set  
\begin{equation}\label{eq3.A}
\wh X_N(\la):=(q^{\la_1}t^{N-1}s, \,q^{\la_2}t^{N-2}s, \,\dots,\, q^{\la_N}s).
\end{equation} 
Then
$$
\bar P^{*(N)}_\mu(\wh X_N(\la); q,t,s)=0 \quad \text{if $\la\ne\mu$ and $|\la|\le|\mu|$},
$$ 
and
$$
\bar P^{*(N)}_\mu(\wh X_N(\mu); q,t,s)\ne0.
$$

$\bullet$ \emph{Normalization}: the monomial $z_1^{\la_1}\dots z_N^{\la_N}$ enters $\bar P^{*(N)}_\mu(\ccdot;q,t, s)$ with coefficient equal to $1$.

\smallskip

The polynomials $\bar P^{*(N)}_\mu(\ccdot;q,t, s)$ are called the \emph{$BC$-type interpolation polynomials}.  

The vanishing property in fact holds in a stronger version:

$\bullet$ \emph{Extra vanishing property}:  $\bar P^{*(N)}_\mu(\wh X_N(\la); q,t,s)=0$ whenever the diagram $\la$ does not contain the diagram $\mu$.

\smallskip

One more important result is the following.

\smallskip

$\bullet$ \emph{Combinatorial formula}:
\begin{multline}\label{eq3.B}
\bar P^{*(N)}_\mu(z_1,\dots,z_N;q,t,s)=\sum_{T\in\RTab(\mu,N)}\psi_T(q,t)\\
\times\prod_{(i,j)\in\mu}\left(z_{T(i,j)}+\dfrac1{z_{T(i,j)}}-q^{j-1}t^{N-i+1-T(i,j)}s- \dfrac1{q^{j-1}t^{N-i+1-T(i,j)}s}\right).
\end{multline}

All these results were first established by Okounkov \cite{Ok-TG-1998}. Later on, in \cite{Ok-TG-2003}, he found another proof of the combinatorial formula. Rains \cite{Rains} proposed an alternative approach to the $BC$-type interpolation polynomials; it relies on the existence of a certain partial $q$-difference operator, which acts both on the arguments $z_i$ and the parameter $s$, and has the interpolation polynomials as eigenfunctions. See Section 3 in \cite{Rains}. 

The notations in these three papers are different. In particular, \cite{Ok-TG-1998} deals with `shifted symmetric' polynomials. The symmetric polynomials $\bar P^{*(N)}_\mu(z_1,\dots,z_N;q,t,s)$ from \cite{Rains} are related to shifted symmetric polynomials $P^*_\mu(\ccdot;q,t,s)$ from \cite{Ok-TG-1998} as follows:
$$
\bar P^{*(N)}_\mu(z_1,\dots,z_N;q,t,s)=(t^{N-1}s)^{|\mu|}P^*_\mu(z_1 t^{1-N}s^{-1}, z_2 t^{2-N} s^{-1},\dots, z_N s^{-1};q,t,s).
$$

\subsection{Koornwinder polynomials} 
The \emph{Koornwinder polynomials} form a six-pa\-ra\-meter family of bases in  the space of $BC_N$-invariant Laurent polynomials. These polynomials, introduced by Koornwinder \cite{Koo},  are a multidimensional analogue of the Askey--Wilson polynomials. For our purpose, the most convenient reference is Rains' paper \cite{Rains}, and we keep to his notation. 

The six parameters are the Macdonald parameters $q$ and $t$, and four extra parameters $t_0,t_1,t_2,t_3$. The corresponding Koornwinder polynomials are denoted as 
$K_\la^{(N)}(\ccdot;q,t;t_0,t_1,t_2,t_3)$, where the index $\la$ ranges over $\Y(N)$. See \cite[Section 5, Definition 2]{Rains}. 

A fundamental result is Okounkov's binomial formula \cite{Ok-TG-1998} which describes the expansion of the Koornwinder polynomials in the $BC$-type interpolation polynomials. We will use it in the proof of Proposition \ref{prop3.A} below. 

\subsection{Multivariable big $q$-Jacobi polynomials $\varphi_{\la\mid N}(\ccdot; q, t; a,b,c,d)$}
The classical (i.e. univariate) big $q$-Jacobi polynomials are described in Andrews--Askey in \cite{AA-1985}. They are expressed through the basic hypergeometric series ${}_3\phi_2$ and can be obtained as a degeneration of the Askey--Wilson polynomials (the latter are expressed through ${}_4\phi_3$), see Koekoek--Swarttouw  \cite{KS}, Koornwinder \cite{Koo-2014}.

The multivariable big $q$-Jacobi polynomials were constructed by Stokman \cite{St}. In Stokman--Koornwinder \cite{StK} it was shown that they can be obtained from the Koornwinder polynomials by a limit transition, similar to that in the one-dimensional case.  

For our purpose it is more convenient to introduce the multivariable big $q$-Jacobi polynomials directly as a degeneration of Koornwinder polynomials, relying on the works of Okounkov \cite{Ok-TG-1998} and Rains \cite{Rains}. The result is stated in the two propositions below.

\begin{proposition}\label{prop3.A}
Let $N=1,2,\dots$ and\/ $\la\in\Y(N)$. Next, let $(a,b,c,d)$ be a quadruple of parameters such that
$$
b<0<a, \qquad c=\bar d\in\C\setminus\R.
$$
There exists a limit 
\begin{multline}\label{eq3.C}
\varphi_{\la\mid N}(x_1,\dots,x_N; q, t; a,b,c,d)\\
:=\lim_{\eps\to0}\eps^{|\la|}K^{(N)}_\la\left(x_1\eps^{-1},\dots,x_N\eps^{-1}; q,t; c\eps, d\eps, \frac q{a\eps}, \frac q{b\eps}\right).
\end{multline}
\end{proposition}

Here the limit is understood as convergence in a large enough finite-dimensional subspace of the space $\C[z_1^{\pm1},\dots,z_N^{\pm1}]$ of Laurent polynomials (for instance, one can take the linear span of monomials $z_1^{\ell_1}\dots z^{\ell_N}$ such that $|\ell_i|\le\la_1$ for each $i=1,\dots,N$). In the limit regime indicated in \eqref{eq3.C}, all monomials entering the Koornwinder polynomials and containing negative powers are killed and so Laurent polynomials turn into ordinary polynomials. 

Proposition \ref{prop3.A} agrees with item (1) of \cite[Theorem 5.1]{StK}, only the notation of the papers \cite{St}, \cite{StK} is different from ours. The correspondence between our parameters $(\abcd)$ and those in \cite{St}, \cite{StK} (which we rename to $(A,B,C,D)$ in order to avoid confusion) is the following:
\begin{equation}\label{eq3.Stokman}
\begin{gathered}
A=ca^{-1}, \quad B=db^{-1}, \quad C=a^{-1}q, \quad D=-b^{-1}q,\\
a=q/C, \quad b=-q/D, \quad c=Aq/C, \quad d=-Bq/D.
\end{gathered}
\end{equation}

We prove Proposition \ref{prop3.A} together with the following result, where we make $(a,b,c,d)$ depending on $N$. 

\begin{proposition}\label{prop3.B}
Fix a quadruple $(\alde)$ such that
$$
\be<0<\al, \qquad \ga=\bar \de\in\C\setminus\R.
$$
For $N=1,2,\dots$ and $\la\in\Y(N)$ the following expansion holds
\begin{multline}\label{eq3.D}
\varphi_{\la\mid N}(x_1,\dots,x_N; q, t; \al,\be,\ga t^{1-N}, \de t^{1-N})\\
=\sum_{\mu\subseteq\la}\frac{(t^N;q,t)_\la}{(t^N;q,t)_\mu}\rho(\la,\mu;q,t;\alde)\frac{I_{\mu\mid N}(\ga x_1,\dots,\ga x_N; q,t)}{\ga^{|\mu|}},
\end{multline}
where the coefficients $\rho(\la,\mu;q,t;\alde)$ do not depend on $N$ and are given by formula \eqref{eq3.G} below.
\end{proposition} 

In the special case $t=q$ a similar claim was proved in \cite[Proposition 3.2]{Ols-FAA-2017} (then the coefficients are given by a determinantal formula). 

\begin{proof}[Proof of the propositions]
We begin with  Okounkov's binomial formula \cite[Theorem 7.1]{Ok-TG-1998}, which we write in Rains' notation (see Theorem 5.1 in \cite{Rains} and Definition 2 before it):
\begin{multline}\label{eq3.H}
K^{(N)}_\la(z_1,\dots,z_N;q,t;t_0,t_1,t_2,t_3)\\
=\sum_{\mu\subseteq\la}\begin{bmatrix}\la\\ \mu\end{bmatrix}_{q,t,t^{N-1}\wh t_0}\dfrac{k^0_\la(q,t,t^N;t_0:t_1,t_2,t_3)}{k^0_\mu(q,t,t^N;t_0:t_1,t_2,t_3)} \bar P^{*(N)}_\mu(z_1,\dots,z_N;q,t,t_0).
\end{multline}

Let us explain the notation. The first term on the right-hand side is defined by
\begin{equation}\label{eq3.E}
\begin{bmatrix}\la\\ \mu\end{bmatrix}_{q,t,s}:=\frac{\bar P^{*(N)}_\mu(\wh X_N(\la);q,t,t^{1-N}s)}{\bar P^{*(N)}_\mu(\wh X_N(\mu);q,t,t^{1-N}s)},
\end{equation}
where $\wh X_N(\ccdot)$ was defined in \eqref{eq3.A} (see the beginning of section 4 in \cite{Rains}). The quantity \eqref{eq3.E} is a kind of binomial coefficient. As is seen from \eqref{eq3.A} and the combinatorial formula \eqref{eq3.B}, this `binomial coefficient' is invariant under the change $s\mapsto -s$ (because under this transformation both the numerator and the denominator in \eqref{eq3.E} are multiplied by $(-1)^{|\mu|}$). Another important remark is that for $q,t,s$ fixed, \eqref{eq3.E} does not depend on $N$, provided that $N\ge\ell(\la)$: indeed, this holds both for the numerator and denominator in \eqref{eq3.E}, as is seen from \cite[Lemma 3.6]{Rains}.

The quantity $\wh t_0$ is defined by
\begin{equation}\label{eq3.F}
\wh t_0:=\sqrt{t_0t_1t_2t_3/q},
\end{equation}
where the choice of the root does not matter, because \eqref{eq3.E} is invariant under $s\mapsto -s$.

Next, 
\begin{multline*}
k^0_\la(q,t,t^N;t_0:t_1,t_2,t_3)\\
:=(t^{N-1} t_0)^{-|\la|}t^{n(\la)}\dfrac{(t^N;q,t)_\la (t^{N-1}t_0t_1;q,t)_\la (t^{N-1}t_0t_2;q,t)_\la (t^{N-1}t_0t_3; q,t)_\la}{C^-_\la(t;q,t)C^+_\la (t^{2N-2}\wh t_0\,\!^2;q,t)},
\end{multline*}
where (see \cite[Section 2]{Rains})
\begin{equation}\label{eq3.M1}
C^-_\la(z;q,t):=\prod_{(i,j)\in\la}(1-zq^{\la_i-j}t^{\la'_j-i})
\end{equation}
and
\begin{equation}\label{eq3.M2}
C^+_\la(z;q,t):=\prod_{(i,j)\in\la}(1-zq^{\la_i+j-1}t^{2-\la'_j-i}).
\end{equation}

Now we specialize
$$
t_0=\ga t^{1-N}\eps, \quad t_1=\de t^{1-N}\eps, \quad t_2=\frac{q}{\al\eps}, \quad t_3=\frac{q}{\be\eps}.
$$
Then
\begin{gather*}
t^{N-1}\wh t_0=\sqrt{\ga\de q/(\al\be)},\\
t^{N-1}t_0=\ga\eps,\\
(t^{N-1}t_0t_1;q,t)_\la=(\ga\de t^{1-N}\eps^2;q,t)_\la=1+O(\eps^2),\\
(t^{N-1}t_0t_2;q,t)_\la=(\ga q/\al;q,t)_\la,\\
(t^{N-1}t_0 t_3;q,t)_\la=(\ga q/\be;q,t)_\la,\\
C^+_\la(t^{2N-2}\wh t_0\,\!^2;q,t)=C^+_\la(\ga\de q/(\al\be);q,t).
\end{gather*}

Therefore, the coefficient of $\bar P^{*(N)}_\mu(z_1,\dots,z_N;q,t,t_0)$ in \eqref{eq3.H} can be written as
\begin{equation}\label{eq3.I}
\frac{(t^N;q,t)_\la}{(t^N;q,t)_\mu}(1+O(\eps^2))\eps^{|\mu|-|\la|}\rho(\la,\mu;q,t;\alde),
\end{equation}
where
\begin{multline}\label{eq3.G}
\rho(\la,\mu;q,t;\alde):=\ga^{|\mu|-|\la|}t^{n(\la)-n(\mu)}\begin{bmatrix}\la\\ \mu\end{bmatrix}_{q,t,\sqrt{\ga\de q/(\al\be)}}\\
\times\frac{(\ga q/\al;q,t)_\la (\ga q/\be;q,t)_\la}{(\ga q/\al;q,t)_\mu (\ga q/\be;q,t)_\mu}\, \frac{\prod\limits_{(i,j)\in\mu}(1-q^{\mu_i-j}t^{\mu'_j-i+1})}{\prod\limits_{(i,j)\in\la}(1-q^{\la_i-j}t^{\la'_j-i+1})}\, \dfrac{\prod\limits_{(i,j)\in\mu}\left(1-\frac{\ga\de}{\al\be}q^{\mu_i+j}t^{2-\mu'_j-i}\right)}{\prod\limits_{(i,j)\in\la}\left(1-\frac{\ga\de}{\al\be}q^{\la_i+j}t^{2-\la'_j-i}\right)}.
\end{multline}

Note that this expression is well defined. Indeed, due to the constraints on $(\alde)$, the numbers  $\ga/\al$ and $\ga/\be$ are not real, so all the factors in the first fraction are nonzero. Next, for any box $(i,j)\in\la$ one has $\la_i-j\ge0$ and $\la'_j-i\ge0$, so that $q^{\la_i-j}t^{\la'_j-i+1}<1$, and hence the denominator of the second fraction is nonzero. Finally, since $\frac{\ga\de}{\al\be}<0$, the same holds for the third fraction, too. 

Now we examine the limit behavior of $\bar P^{*(N)}_\mu(z_1,\dots,z_N;q,t,t_0)$ specialized at 
$$
z_1:=x_1\eps^{-1}, \; \dots,\; z_N:=x_N\eps^{-1}, \qquad t_0:=\ga t^{1-N}\eps.
$$
Look at the combinatorial formula \eqref{eq3.B}. After this specialization, the $(i,j)$-factor in the second line of \eqref{eq3.B} takes the form
\begin{multline*}
x_{T(i,j)}\eps^{-1}+\dfrac{\eps}{x_{T(i,j)}} -q^{j-1}t^{-T(i,j)-i+2}\ga\eps-\dfrac{\ga^{-1}\eps^{-1}}{q^{j-1}t^{-T(i,j)-i+2}}\\
=\ga^{-1}\eps^{-1}(\ga x_{T(i,j)}-q^{1-j}t^{T(i,j)+i-2}+O(\eps^2)).
\end{multline*}
Comparing this to the combinatorial formula \eqref{eq2.A} we see that
\begin{equation}\label{eq3.J}
\lim_{\eps\to0}\eps^{|\mu|}\bar P^{*(N)}_\mu(x_1\eps^{-1},\dots,x_N\eps^{-1};q,t,t^{1-N}\ga \eps)=\ga^{-|\mu|}I_{\mu\mid N}(\ga x_1,\dots,\ga x_N;q,t).
\end{equation}

Collecting \eqref{eq3.H}, \eqref{eq3.I}, and \eqref{eq3.J} we conclude that 
\begin{multline*}
\lim_{\eps\to0}\eps^{|\la|}K^{(N)}_\la\left(x_1\eps^{-1},\dots,x_N\eps^{-1}; q,t; \ga t^{1-N}\eps, \de t^{1-N}\eps, \frac q{\al\eps}, \frac q{\be\eps}\right)\\
=\sum_{\mu\in\Y(N)}\frac{(t^N;q,t)_\la}{(t^N;q,t)_\mu} \rho(\la,\mu;q,t;\alde)\frac{I_{\mu\mid N}(\ga x_1,\dots,\ga x_N;q,t)}{ \ga^{|\mu|}}.
\end{multline*}

This completes the proof of \eqref{eq3.C} and \eqref{eq3.D}.
\end{proof}

\begin{corollary}\label{cor3.A}
The following expansion holds{\rm:}
\begin{multline}\label{eq3.L}
\varphi_{\la\mid N}(x_1,\dots,x_N;q,t;\al,\be,\ga t^{1-N},\de t^{1-N})
=\sum_{\nu\subseteq\la}\frac{(t^N;q,t)_\la}{(t^N;q,t)_\nu}\\
\times\pi(\la,\nu;q,t;\alde) P_{\nu\mid N}(x_1,\dots,x_N;q,t),
\end{multline}
where the coefficients $\pi(\la,\nu;q,t;\alde)$  do not depend on $N$ and are given by  
\begin{equation}\label{eq3.L1}
\pi(\la,\nu;q,t;\alde)=\sum_{\mu:\, \nu\subseteq\mu\subseteq\la} \rho(\la,\mu;q,t;\alde)\si(\mu,\nu;q,t)\ga^{|\nu|-|\mu|},
\end{equation}
where, in turn, the coefficients on the right-hand side of \eqref{eq3.L1} are given by 
\eqref{eq3.G} and \eqref{eq2.C}. 
\end{corollary}

\begin{remark}\label{rem3.A}
Rewrite \eqref{eq3.L} as
\begin{multline}\label{eq3.stability}
\frac{\varphi_{\la\mid N}(x_1,\dots,x_N;q,t;\al,\be,\ga t^{1-N},\de t^{1-N})}{(t^N;q,t)_\la}\\
=\sum_{\nu\subseteq\la}
\pi(\la,\nu;q,t;\alde) \frac{P_{\nu\mid N}(x_1,\dots,x_N;q,t)}{(t^N;q,t)_\nu}.
\end{multline}
The fact that the coefficients in this expansion do not depend on $N$ is of crucial importance for us.
\end{remark}

\begin{proof}
Combining \eqref{eq3.D} with \eqref{eq2.F} we obtain 
\begin{multline*}
\varphi_{\la\mid N}(x_1,\dots,x_N; q, t; \al,\be,\ga t^{1-N}, \de t^{1-N})\\
=\sum_{\mu\subseteq\la}\frac{(t^N;q,t)_\la}{(t^N;q,t)_\mu}\rho(\la,\mu;q,t;\alde)\frac{I_{\mu\mid N}(\ga x_1,\dots,\ga x_N; q,t)}{\ga^{|\mu|}}\\
=\sum_{\mu\subseteq\la}\frac{(t^N;q,t)_\la}{(t^N;q,t)_\mu}\rho(\la,\mu;q,t;\alde)\sum_{\nu\subseteq\mu}\frac{(t^N;q,t)_\mu}{(t^N;q,t)_\nu}\si(\mu,\nu;q,t)\frac{P_{\nu\mid N}(\ga x_1,\dots,\ga x_N; q,t)}{\ga^{|\mu|}}.
\end{multline*} 
Because $P_{\nu\mid N}(\ccdot;q,t)$ is homogeneous of degree $|\nu|$, we have
$$
\frac{P_{\nu\mid N}(\ga x_1,\dots,\ga x_N; q,t)}{\ga^{|\mu|}}=\ga^{|\nu|-|\mu|}P_{\nu\mid N}(x_1,\dots,x_N; q,t).
$$
This gives the desired expression.
\end{proof}

\begin{corollary}\label{cor3.B}
In the expansion \eqref{eq3.L}, the coefficient with $\nu=\la$ equals $1$, so that 
\begin{multline}\label{eq3.K}
\varphi_{\la\mid N}(x_1,\dots,x_N;q,t;\al,\be,\ga t^{1-N},\de t^{1-N})\\
=P_{\la}(x_1,\dots,x_N;q,t)+\text{\rm lower degree terms}.
\end{multline}
\end{corollary}

\begin{proof}
By virtue of \eqref{eq3.L1},
\begin{equation}\label{eq3.K1}
\pi(\la,\la;q,t;\alde)=\rho(\la,\la;q,t;\alde)\si(\la,\la;q,t).
\end{equation}
Next, from the explicit formula \eqref{eq3.G} it is seen that the coefficient $\rho(\la,\mu;q,\alde)$ equals $1$ for $\mu=\la$. Likewise, the coefficient $\si(\mu,\nu;q,t)$ equals $1$ for $\mu=\nu$ (see \eqref{eq2.E}). This implies the desired claim. 
\end{proof} 

\begin{proposition}\label{prop3.C}
{\rm(i)} The polynomials $\varphi_{\la\mid N}(x_1,\dots,x_N;q,t;\al,\be,\ga t^{1-N},\de t^{1-N})$  are symmetric with respect to the transpositions $(\al,\be)\to(\be,\al)$ and $(\ga,\de)\to(\de,\ga)$.

{\rm(ii)} The same holds for the coefficients $\pi(\la,\nu;q,t;\alde)$ of the expansion \eqref{eq3.L}.
\end{proposition}

\begin{proof}
The Koornwinder polynomials are invariant under permutations of all four parameters $(t_0,t_1,t_2,t_3)$, see \cite[Theorem 5.5]{Rains}. By virtue of \eqref{eq3.C}, this implies claim (i). Next, claim (i) evidently implies claim (ii). 
\end{proof}

Note that the symmetry under transposition $\ga\leftrightarrow\de$ is not evident from the expansion \eqref{eq3.D} of Proposition \ref{prop3.B}.

\section{One-particle hypergeometric ensembles}\label{sect4}

Recall the definition of the two-sided $q$-lattice $\L$:
$$
\L:=\L_+\sqcup \L_-\subset\R^*, \qquad \L_\pm:=\zeta_\pm q^\Z, \quad 0<q<1,
$$
where $\zeta_+>0$ and $\zeta_-<0$ are arbitrary fixed parameters. 

We will deal with the following probability distributions on $\L$: 
\begin{equation}\label{eq4.D}
M(x;q;a,b,c,d)=\frac1{Z(q;a,b,c,d)}\, W(x;q;a,b,c,d), \qquad x\in\L,
\end{equation}
where $x\mapsto W(x;q;a,b,c,d)$ is the nonnegative summable function on $\L$ of the `hypergeometric' form
\begin{equation}\label{eq4.A}
W(x;q;a,b,c,d):=(1-q)|x|\frac{(ax;q)_\infty(bx;q)_\infty}{(cx;q)_\infty(dx;q)_\infty}, \qquad x\in\L,
\end{equation}
depending (apart from $q$) on a quadruple $(a,b,c,d)$ of complex parameters, and $Z(q;a,b,c,d)$ is the normalization constant, 
\begin{equation}\label{eq4.C}
Z(q;a,b,c,d):=\sum_{x\in\L}W(x;q;a,b,c,d).
\end{equation}

Here is a (non-complete) description of the values of $(a,b,c,d)$ for which the above definition makes sense. We divide them into three `series', which we call principal, complementary, and degenerate . 

\smallskip

(1) \emph{Principal series}: $c=\bar d\in\C\setminus\R$; $a=\bar b\in\C\setminus\R$; $ab<cdq$.

\smallskip

(2) \emph{Complementary series}: $c=\bar d\in\C\setminus\R$, both $a$ and $b$ lie in an open interval of the form $(\zeta^{-1}_+ q^n, \zeta^{-1}_+ q^{n-1})$ or $(\zeta^{-1}_- q^{n-1}, \zeta^{-1}_-q^n)$ with $n\in\Z$; $ab<cdq$.

\smallskip

(3) \emph{Degenerate series}: $c=\bar d\in\C\setminus\R$; $b<0<a$ and $a\in\zeta^{-1}_+ q^\Z$, $b\in\zeta_-^{-1}q^\Z$. 

\smallskip

The constraint imposed on $(c,d)$ guarantees that the denominator in \eqref{eq4.A} is strictly positive for all $x\in\L$. 

In the case (1), the constraint on $(a,b)$  implies that $(ax;q)_\infty$ and $(bx;q)_\infty$ are nonzero and complex conjugate for any $x\in\L$, which implies that the numerator in \eqref{eq4.A} is strictly positive for all $x\in\L$. 

In the case (2), the same positivity property holds, because then $(ax;q)_\infty$ and $(bx;q)_\infty$ are nonzero reals of the same sign. 

In the case (3),  the numerator is strictly positive if $b^{-1}q\le x\le a^{-1}q$, and vanishes otherwise.  Thus, the support of the measure is the following proper subset of $\L$:
$$
\L[b^{-1}q,  a^{-1}q]:=\{x\in\L: b^{-1}q\le x\le a^{-1}q\}.
$$

Note that the set of admissible values of parameters $c$ and $d$ can be extended --- they may be taken real and satisfying the same type of constraint as for $(a,b)$ in the case (2), but we will not strive for maximum generality. 

Due to the factor $|x|$ in \eqref{eq4.A}, the function $W(x;a,b,c,d)$ is summable near zero. In the case (3), when the support of the function is bounded away from $\pm\infty$, this is enough to conclude that the function is entirely summable.  

To verify the summability of the function $W(x;q;a,b,c,d)$ in the cases (1) and (2) we need the following lemma (cf. Groenevelt \cite[Lemma 2.4]{Gro-2009}).  

\begin{lemma}\label{lemma4.A}
As above, fix $c$ and $d$ such that $c=\bar d\in\C\setminus\R$. Next, fix arbitrary $R_1>0$,  $R_2>0$, and $\zeta\in\R^*$. Then for any complex $a$ and $b$ such that $|a|\le R_1$, $|b|\le R_2$ the following estimate holds
\begin{equation}\label{eq4.B}
\left|\frac{(a\zeta q^{-n};q)_\infty(b\zeta q^{-n};q)_\infty}{(c\zeta q^{-n};q)_\infty(d\zeta q^{-n};q)_\infty}\right|\le A(c,d, R_1,R_2, \zeta) \left(\frac{R_1R_2}{cd}\right)^n, \qquad n\in\Z_{\ge1}, 
\end{equation}
with some positive constant $A(c,d,R_1,R_2,\zeta)$ which does not depend on $a$, $b$, and $n$. 
\end{lemma}

Note that the estimate \eqref{eq4.B} is uniform on $(a,b)$ ranging in the bidisc $|a|\le R_1$, $|b|\le R_2$. This fact is not essential in the context of the present section, but it is important for justifying the procedure of analytic continuation in subsection \ref{sect6.3}. 

\begin{proof}
Due to the constraint on $(c,d)$ and because $\zeta$ is real, the denominator on the left-hand side is strictly positive. Without loss of generality we may assume that $\zeta>0$. Then we have
$$
\left|\frac{(a\zeta q^{-n};q)_\infty(b\zeta q^{-n};q)_\infty}{(c\zeta q^{-n};q)_\infty(d\zeta q^{-n};q)_\infty}\right| \le
\frac{(-R_1\zeta q^{-n};q)_\infty(-R_2\zeta q^{-n};q)_\infty}{(c\zeta q^{-n};q)_\infty(d\zeta q^{-n};q)_\infty}.
$$
Next, the right-hand side is further transformed to
\begin{gather*}
\frac{(-R_1\zeta;q)_\infty(-R_2\zeta;q)_\infty}{(c\zeta;q)_\infty(d\zeta;q)_\infty} \cdot
\frac{(-R_1\zeta q^{-n};q)_n(-R_2\zeta q^{-n};q)_n}{(c\zeta q^{-n};q)_n(d\zeta q^{-n};q)_n}\\
=\frac{(-R_1\zeta;q)_\infty(-R_2\zeta;q)_\infty}{(c\zeta;q)_\infty(d\zeta;q)_\infty} \cdot
\frac{(-R_1^{-1}\zeta^{-1};q)_n(-R_2^{-1}\zeta^{-1};q)_n}{(c^{-1}\zeta^{-1};q)_n(d^{-1}\zeta^{-1};q)_n}\cdot \left(\frac{R_1R_2}{cd}\right)^n.
\end{gather*}
As $n\to+\infty$, the second fraction on the last line converges and so remains bounded. This gives the desired estimate. 
\end{proof}

\begin{corollary}
In the cases of principal and complementary series the function $x\mapsto W(x;q;a,b,c,d)$ is summable.  
\end{corollary}

\begin{proof}
We have to check that the inequality $ab<cdq$ in the definition of the principal and complementary series ensures the summability  for large values of $|x|$, but this follows directly from the lemma: note that the factor $q$ attached to $cd$ is needed to compensate the growth of the prefactor $|x|$. 
\end{proof}

\begin{remark} 1. The measures $M(\ccdot; a,b,c,d)$ from the principal/complementary series appeared in the works of Askey \cite{Ask-1981}, Groenevelt \cite{Gro-2009}, \cite{Gro-2011}, Groenevelt--Koelink \cite{GroK-2011}.  In particular, Groenevelt \cite{Gro-2011} discovered a link of the measures $M(\ccdot; a,b,c,d)$ with infinite-dimensional unitarizable representations of $\mathcal U_q(su(1,1))$, the   noncompact real form of the quantized universal enveloping algebra $\mathcal U_q(sl(2,\C))$. 

2. The functions $W(\ccdot; a,b,c,d)$ from the degenerate series serve as the weight functions for the univariate big $q$-Jacobi polynomials (Andrews--Askey \cite{AA-1985}, Koekoek--Swarttouw \cite{KS}, Koornwinder \cite{Koo-2014}). 

3. There exists a closed formula for the combinatorial sum \eqref{eq4.C} giving the normalization constant:
\begin{equation}\label{eq4.Z}
Z(q; a,b,c,d)=(1-q)\z_+\, \frac{(q,a/c,a/d,b/c,b/d;q)_\infty}{(ab/(cdq);q)_\infty}\,
\frac{\theta_q(\z_-/\z_+,cd\z_-\z_+)}{\theta_q(c\z_-,d\z_-,c\z_+,d\z_+)}\,,
\end{equation}
where the theta function $\theta_q(u)$ of a single variable $u$ is defined by
$$
\theta_q(u)=(u, q/x;q)_\infty
$$
and we use the shorthand notation
$$
(u_1,\dots,u_m;q)_n=(u_1;q)_n\dots(u_m;q)_n, \qquad \theta_q(u_1,\dots,u_m):=\theta_q(u_1)\cdots\theta_q(u_m).
$$ 
See Koornwinder  \cite{Koo-2014}, formula (139) and subsequent comments. The sum \eqref{eq4.C} can be interpreted as a $q$-integral; in this form it was evaluated by Askey \cite[(3.13)]{Ask-1981} in the case $\zeta_+=1$, $\zeta_-=-1$.

\end{remark}

\section{$N$-particle hypergeometric ensembles}\label{sect5}

Throughout this section $q$ and $t$ are fixed parameters in $(0,1)$ and $\tau:=\log_q t$. We also fix two extra parameters $\zeta_+>0$, $\zeta_-<0$. 

\subsection{Particle configurations}\label{sect5.1}

Let $\R_{>0}$ and $\R_{<0}$ be the sets of strictly positive and strictly negative reals, respectively, and $\R^*:=\R_{>0}\sqcup\R_{<0}$. Next, $\Z_{\ge0}$ and $\Z_{\ge1}$ is our notation for the nonnegative integers and strictly positive integers, respectively. 

Let $\wt\Om$ denote the set of finite or countable subsets $X\subset\R^*$ of the form $X=X^+\cup X^-$ with 
\begin{equation}\label{eq5.A1}
 X^+:=X\cap\R_{>0}=\{x^+_1,x^+_2,\dots\}, \quad X^-:=X\cap\R_{<0}=\{x^-_1,x^-_2,\dots\},
\end{equation}
where each of the sequences $\{x^+_i\}$, $\{x^-_i\}$ may be infinite, finite or void, and
$$
x^\pm_1\in \zeta_{\pm} q^\Z, \qquad x^\pm_{i+1}\in x^\pm_{i} t q^{\Z_{\ge0}} \quad \text{for $i=1,2,\dots$}\,. 
$$
This means, in particular, that
$$
x^-_1<x^-_2<\dots <0<\dots <x^+_2<x^+_1
$$
and 
$$
\log_q x^+_{i+1}-\log_q x^+_{i}\ge\tau, \quad \log_q |x^-_{i+1}|-\log_q |x^-_{i}|\ge\tau, \qquad i=1,2,\dots\,.
$$
Note that in the special case of $\tau\in\Z_{\ge1}$ this definition agrees with that in subsection \ref{sect1.1.1}.  

The topology in $\wt\Om$ is introduced through a uniform structure. The latter is defined by proclaiming two configurations to be \emph{$\eps$-close}  if they coincide outside the interval $(-\eps,\eps)$. In this topology, $\wt\Om$ is a locally compact space with the stratification
$$
\wt\Om=\Om_\infty \sqcup \bigsqcup_{N=0}^\infty \Om_N.
$$
where 
$$
\Om_N:=\{X\in\wt\Om: |X|=N\}, \quad \Om_\infty:=\{X\in\wt\Om: |X|=\infty\}.
$$

\subsection{Degenerate series of hypergeometric measures}

By definition, these are the orthogonality measures of $N$-variate big $q$-Jacobi polynomials. Those measures were found by Stokman \cite{St}, and below we restate his definition in our notation. 

\subsubsection{A $(q,t)$-version of $\prod\limits_{1\le i<j\le N}(x_i-x_j)^\be$}

For $X\in\Om_N$ we set
\begin{equation}\label{eq5.D}
V_{q,t}(X)=\prod_{\substack{x,y\in X\\ y<x}}(x-y)\cdot \prod_{\substack{x\in X^+\\ y<x}}x^{2\tau-1}\frac{(qt^{-1}yx^{-1};q)_\infty}{(tyx^{-1};q)_\infty} \cdot \prod_{\substack{x\in X^-\\ y<x}}|y|^{2\tau-1}\frac{(qt^{-1}xy^{-1};q)_\infty}{(txy^{-1};q)_\infty},
\end{equation}
where $X^\pm$ are defined by \eqref{eq5.A1} (cf. \cite[(5.4)]{St}). 

From the definition of $\Om_N$ it is seen that all $(\ccdot;q)_\infty$ factors on the right-hand side are of the form  $(u;q)_\infty$ with $u<1$ and hence are strictly positive. Therefore, $V_{q,t}(X)>0$ for all $X\in\Om_N$.  
In the special case $\tau\in\Z_{\ge1}$ the expression on the right-hand side simplifies and reduces to \eqref{eq1.E}.

\subsubsection{The function $W(x;q; a,b,c,d)$}\label{sect5.1.2}

Let us fix an additional quadruple of nonzero parameters $(a,b,c,d)$ such that
\begin{equation}\label{eq5.C}
a\in \zeta_+^{-1} q^\Z, \quad b\in\zeta_-^{-1} q^\Z, \quad c=\bar d\in\C\setminus\R.
\end{equation}
Note that $b<0<a$. Recall that our notation for these parameters differs from that of \cite{St}, see \eqref{eq3.Stokman} above.

The expression \eqref{eq4.A} for the function $W(x;q;a,b,c,d)$ is well defined for any real $x$. It is strictly positive if $b^{-1}q \le x \le a^{-1}q$ and vanishes at the points 
$$
a^{-1}, a^{-1}q^{-1}, a^{-1}q^{-2}, \dots; \quad b^{-1}, b^{-1}q^{-1}, b^{-1}q^{-2}, \dots\,.
$$

 \subsubsection{Some constant factors}
Following  \cite[(5.5)]{St} we set
\begin{equation}\label{eq5.A}
C_N(k;\zeta_+,\zeta_-):=\prod_{\substack{1\le i<j\le N\\ i\le k}}\psi_\tau(\zeta_-\zeta_+^{-1}t^{N-i-j+1}), \qquad 0\le k\le N,
\end{equation}
where
\begin{equation}\label{eq5.L}
\psi_\tau(u):=(-u)^{2\tau-1}\frac{\th_q(tu)}{\th_q(t u^{-1})}, \qquad u<0,
\end{equation}
and 
\begin{equation}\label{eq5.M}
\th_q(v):=(v;q)_\infty(qv^{-1};q)_\infty.
\end{equation}
An important property of $\psi_\tau(u)$ is that it is a \emph{quasi-constant function} meaning that
\begin{equation}\label{eq5.N}
\psi_\tau(qu)=\psi_\tau(u).
\end{equation}

Because $\zeta_-\zeta_+^{-1}<0$, the expression on the right-hand side of \eqref{eq5.A} makes sense and is strictly positive.  

In the special case of $\tau\in\Z_{\ge1}$ the function $\psi_\tau(u)$ on $\R_{<0}$ is identically equal to $1$, so that the constants \eqref{eq5.A} also equal $1$. 

\subsubsection{The degenerate series measures $M_N(\ccdot;q,t;a,b,c,d)$}

Let $\Om_N[b^{-1}q,a^{-1}q]$ denote the subset of $\Om_N$ formed by the configurations contained in the closed interval $[b^{-1}q,a^{-1}q]$. Observe that the product $\prod_{x\in X} W(x;q; a,b,c,d)$ vanishes unless $X\in\Om_N[b^{-1}q,a^{-1}q]$. Indeed, if $X\notin\Om_N[b^{-1}q,a^{-1}q]$, then at least one of the two inclusions 
$$
x^+_1\in \{a^{-1}, a^{-1}q^{-1}, a^{-1}q^{-2}, \dots\}, \quad x^-_1\in\{b^{-1}, b^{-1}q^{-1}, b^{-1}q^{-2}, \dots\}
$$
holds, and then the corresponding factor $W(x^\pm_1;q;a,b,c,d)$ vanishes. 

In accordance with \cite[sect. 5]{St} we set for $X\in\Om_N$
\begin{equation}\label{eq5.B}
\begin{gathered}
\wt M_N(X;q,t;a,b,c,d)=C_N(|X^+|;\zeta_+,\zeta_-)\,V_{q,t}(X)\prod_{x\in X} W(x; q; a,b,c,d), \\
M_N(X;q,t;a,b,c,d):=(Z_N(q,t;a,b,c,d))^{-1} \wt M_N(X;q,t;a,b,c,d),
\end{gathered}
\end{equation}
where $Z_N(q,t;a,b,c,d)$ is the normalization factor, 
\begin{equation}\label{eq5.Z_N}
Z_N(q,t;a,b,c,d):=\sum_{X\in\Om_N}\wt M_N(X;q,t;a,b,c,d).
\end{equation}
The  summation here is taken in fact over $X\in\Om_N[b^{-1}q,a^{-1}q]$. From this and the bound on $V_{q,t}(X)$ established below in Lemma \ref{lemma5.A} it follows that the series converges. 

Thus, $M_N(\ccdot;q,t;a,b,c,d)$ is a probability measure on $\Om_N[b^{-1}q,a^{-1}q]$.

\begin{remark}\label{rem5.A}
The combinatorial sum on the right-hand side of \eqref{eq5.Z_N} can be written as a linear combination of $N$-fold $q$-integrals. A remarkable fact is that its value $Z_N(q,t;a,b,c,d)$ is given by a closed multiplicative formula. This result is contained in the works of  Stokman \cite[Corollary 7.6]{St-2000}, Tarasov--Varchenko \cite[Theorem (E.10)]{TV},  and Ito--Forrester \cite[Corollary 4.3]{IF}. It is a nontrivial generalization of Evans' result mentioned in the introduction (Remark \ref{rem1.A}). 
\end{remark}

\subsubsection{Orthogonality of multivariable big $q$-Jacobi polynomials}

We consider the polynomials $\varphi_{\la\mid N}(\ccdot; q,t;a,b,c,d)$ defined in \eqref{eq3.C}.

\begin{theorem}\label{thm5.A}
For each $N\in\Z_{\ge1}$, the measure $M_N(\ccdot;q,t;a,b,c,d)$ defined by \eqref{eq5.B} serves as an orthogonality measure for  the polynomials $\varphi_{\la\mid N}(\ccdot; q,t;a,b,c,d)$. 
\end{theorem}

This is a reformulation of a claim contained in Stokman \cite[Theorem 5.7 (1)]{St}.
In more detail, the claim is that  for any partitions $\la,\mu\in\Y(N)$
\begin{multline}
\sum_{X\in\Om_N[b^{-1}q, a^{-1}q]}\varphi_{\la\mid N}(X; q,t;a,b,c,d)\varphi_{\mu\mid N}(X; q,t;a,b,c,d)M_N(X;q,t;a,b,c,d)\\
=\begin{cases} \text{\rm a positive constant $h_{\la\mid N}(q,t;a,b,c,d)$}, & \mu=\la,\\
0, & \mu\ne\la.\end{cases}
\end{multline}

\begin{proof}
The limit regime used in our formula \eqref{eq3.C} is taken from \cite{StK}. By virtue of the main result of \cite{StK}, our definition of the multivariable big $q$-Jacobi polynomials agrees with that of \cite{St} (up to notation). Then the desired orthogonality property follows from \cite[Theorem 5.7, item (1)]{St}.  
\end{proof}

As shown in Stokman \cite{St-2000}, the orthogonality measure can also be obtained in another way, via the degeneration of Koornwinder polynomials. 

The core of Stokman's argument in \cite{St} is the claim that $M_N(\ccdot;a,b,c,d)$ serves as the symmetrising measure for a certain second order partial $q$-difference operator (of which the big $q$-Jacobi polynomials are eigenfunctions). The proof of this claim, given in section 6 of \cite{St}, is rather laborious. In the special case of $\tau\in\Z_{\ge1}$ the argument can be simplified.

The following proposition is a complement to Theorem \ref{thm5.A}. 

\begin{proposition}\label{prop5.A}
$M_N(\ccdot;q,t;a,b,c,d)$  is a unique probability measure on $\Om_N[b^{-1}q, a^{-1}q]$ which is orthogonal to polynomials $\varphi_{\la\mid N}(\ccdot;q,t;\abcd)$ with nonzero index $\la$. 
\end{proposition}

\begin{proof}
The closure of $\Om_N[b^{-1}q, a^{-1}q]$ in $\wt\Om$ is the compact set
\begin{equation}
\wt\Om_N[b^{-1}q, a^{-1}q]:=\bigsqcup_{n=0}^N \Om_n[b^{-1}q, a^{-1}q].
\end{equation}
Any polynomial $F\in\Sym(N)$ may be treated as a continuous function on $\wt\Om_N[b^{-1}q, a^{-1}q]$: if necessary, we complement the coordinates by $0$'s. In this way we realize $\Sym(N)$ as a subalgebra in the algebra of real-valued continuous functions on $\wt\Om_N[b^{-1}q, a^{-1}q]$. This subalgebra separates points and contains the constants, hence it is dense in the topology defined by the supremum norm.

Therefore, any finite measure on $\wt\Om_N[b^{-1}q, a^{-1}q]$ (a fortiori on $\Om_N[b^{-1}q, a^{-1}q]$) is uniquely determined by its values on the polynomials from $\Sym(N)$. This implies the desired uniqueness claim. \end{proof}

\subsection{Some estimates}\label{sect5.3}

We use Stokman's formulas from the previous subsection but change the range of the parameters. Our aim is to extend  the definition of the principal series of section \ref{sect4} to the $N$-dimensional case. In this subsection we establish necessary estimates.

\begin{proposition}\label{prop5.B}
We remove the constraints on $(a,b)$ imposed in \eqref{eq5.C} but still assume that $c=\bar d\in\C\setminus\R$. Let $R:=|c|=|d|$.  Let $\wt M_N(X;q,t;\abcd)$ be still defined by formula \eqref{eq5.B}. 

{\rm(i)} The series $\sum_{X\in\Om_N} |\wt M_N(X;q,t;\abcd)|$ 
converges for any $(a,b)\in\C^2$ such that $|ab|<R^2q^{2\tau(N-1)+1}$.

{\rm(ii)} Moreover, for $(c,d)$ fixed, the convergence is uniform on $(a,b)$ ranging in any closed bidisc of the form $|a|\le R_1$, $|b|\le R_2$ such that $R_1>0$, $R_2>0$ and $R_1R_2<R^2q^{2\tau(N-1)+1}$. 
\end{proposition}

First we establish a lemma.    

\begin{lemma}\label{lemma5.A}
The following estimate holds
\begin{equation}\label{eq5.K}
V_{q,t}(X)\le\const \prod_{i=1}^N(\max(1,|x_i|))^{2\tau(N-1)}, \qquad X\in\Om_N.
\end{equation}
\end{lemma}

Here and throughout the proof of the lemma we denote by `$\const$' a varying positive factor which does not depend on $X$, but may depend on $N$, $q$, $t$, $\zeta_+$, and $\zeta_-$. 

\begin{proof}
In the case of $\tau\in\Z_{\ge1}$, when we can use formula \eqref{eq1.E}, the proof is easy. In the general case it is a bit more complicated. 

Recall that $V_{q,t}(X)$ is defined by formula \eqref{eq5.D}. We rewrite its right-hand side as follows
\begin{gather}
\prod_{\substack{x,y\in X\\ y<x}}(x-y)\cdot 
\prod_{\substack{x,y\in X^+\\ y<x}}x^{2\tau-1}\frac{(qt^{-1}yx^{-1};q)_\infty}{(tyx^{-1};q)_\infty} \cdot 
\prod_{\substack{x,y\in X^-\\ y<x}}|y|^{2\tau-1}\frac{(qt^{-1}xy^{-1};q)_\infty}{(txy^{-1};q)_\infty}  \label{eq5.E}   \\
\times \prod_{\substack{x\in X^+\\ -y\in X^-}}x^{2\tau-1}\frac{(-qt^{-1}yx^{-1};q)_\infty}{(-tyx^{-1};q)_\infty}. \label{eq5.F}
\end{gather}

The first product in \eqref{eq5.E}, the Vandermonde, is an alternating sum of monomials in $x_1,\dots,x_N$, of degree at most $N-1$ in each variable. This gives the bound
\begin{equation}\label{eq5.J}
\prod_{\substack{x,y\in X\\ y<x}}(x-y)\le \const \prod_{x\in X}(\max(1,|x|)^{N-1}.
\end{equation}

In the second and third products in \eqref{eq5.E}, each fraction is of the form
\begin{equation}
\frac{(u;q)_\infty}{(v;q)_\infty}, \qquad 0<u\le q, \quad 0<v\le t^2,
\end{equation}
and hence is bounded by a constant. This implies that these two products together are bounded from above by
\begin{equation}\label{eq5.H}
\const \prod_{\substack{x,y\in X^+\\ y<x}}x^{2\tau-1}\cdot 
\prod_{\substack{x,y\in X^-\\ y<x}}|y|^{2\tau-1}.
\end{equation}

We turn now to the last product, that in \eqref{eq5.F}. The previous argument does not work because here we cannot control the size of $yx^{-1}$ --- it may be arbitrarily large.  This difficulty is overcome in the following way: we will  show that the ratio
\begin{equation}\label{eq5.G}
\left(x^{2\tau-1}\frac{(-qt^{-1}yx^{-1};q)_\infty}{(-tyx^{-1};q)_\infty}\right)/
 \left(y^{2\tau-1}\frac{(-qt^{-1}xy^{-1};q)_\infty}{(-txy^{-1};q)_\infty}\right), \quad x\in\ X^+,  -y\in X^-,
\end{equation}
takes only finally many values. 

Suppose this is done (the proof is postponed to the very end). Then we have the freedom to use any of the two expressions in \eqref{eq5.G}, depending on whether $x\ge y$ or $y\ge x$. This implies that the last product is bounded from above by (we again change the notation)
\begin{equation}\label{eq5.I}
\const \prod_{\substack{x\in X^+\\ y\in X^-}}\max(x, |y|)^{2\tau-1}.
\end{equation}

Next, combine the bounds \eqref{eq5.H} and \eqref{eq5.I} together and observe that each letter $x\in X$ occurs there at most $N-1$ times. This shows that all the products together, except the Vandermonde, are bounded by
$$
\const \prod_{x\in X}(\max(1,|x|)^{(2\tau-1)(N-1)}.
$$
Then, taking account of the bound \eqref{eq5.J} for the Vandermonde we finally obtain the desired bound \eqref{eq5.K}. 

It remains to handle the ratio \eqref{eq5.G}. Setting $v=xy^{-1}$ we write it as
$$
v^{2\tau-1}\frac{\th_q(-tv)}{\th_q(-tv^{-1})}=\psi_\tau(-v),
$$
in the notation of \eqref{eq5.L}, \eqref{eq5.M}. Observe now that the possible values of the argument $-v$ are restricted by a subset of the form $\zeta_+\zeta_-^{-1} \{t^m\} q^\Z\subset\R_{<0}$, where $m$ takes only a finite number of values. Because $\psi_\tau$ is a quasi-constant function (see \eqref{eq5.N}), this completes the proof. 
\end{proof}

\begin{proof}[Proof of Proposition \ref{prop5.B}]

For $k=0,\dots,N$, let $\Om_{N,k}\subset\Om_N$ denote the subset of configurations $X$ with $|X^+|=k$. The set $\Om_N$ is the disjoint union of these subsets. Therefore it is enough to establish a uniform estimate for the series
$$
\sum_{X\in\Om_{N,k}} V_{q,t}(X) \prod_{x\in X}|W(x; q;a,b,c,d)|
$$
with $k$ fixed (the constant $C_N(k;\zeta_+,\zeta_-)$ can be ignored). Then $X=(x_1>\dots>x_N)$, where $x_i\in \zeta_i q^\Z$ with some fixed $\zeta_i\in\R^*$. Writing 
$$
X=(x_1,\dots,x_N)=(\zeta_1 q^{n_1},\dots, \zeta_N q^{n_N})
$$
we obtain a bijection between $\Om_{N,k}$ and a subset  $S\subset\Z^N$.  

Suppose now that $(a,b)$ ranges over the bidisk $|a|\le R_1$, $|b|\le R_2$. Taking into account the factor $|x|$ in the definition of the weight function (see \eqref{eq5.A}) we obtain from Lemma \ref{lemma4.A} and Lemma \ref{lemma5.A} the bound
$$
V_{q,t}(X) \prod_{x\in X}|W(x;q; a,b,c,d)|\le \const \prod_{i=1}^N f(n_i), \qquad X\leftrightarrow (n_1,\dots,n_N),
$$
where $f(n)$ is the function on $\Z$ defined by
$$
f(n)=\begin{cases} q^n, & n\ge0\\ \left(\dfrac{R_1R_2}{R^2q^{2\tau(N-1)+1}}\right)^{|n|}, & n<0.\end{cases}
$$
For $R_1, R_2$ fixed the bound is uniform on $(a,b)$ ranging over the corresponding bidisk. 

Since 
$$
\sum_{(n_1,\dots,n_N)\in\Z^N} \prod_{i=1}^Nf(n_i)<\infty \qquad \text{if $R_1R_2<R^2 q^{2\tau(N-1)+1}$},
$$
this proves both claims of the proposition, (i), and (ii). 
\end{proof}

\subsection{Principal series of hypergeometric measures}\label{sect5.4}

Now we can extend the definition of the principal series from section \ref{sect4} to the $N$-dimensional case. We impose the following conditions on the parameters $a,b,c,d$:
\begin{equation}\label{eq5.O}
c=\bar d\in\C\setminus\R, \quad a=\bar b\in\C\setminus\R, \quad ab<cdq^{2\tau(N-1)+1}.
\end{equation}

\begin{proposition}
Under these conditions, the expression \eqref{eq5.B} still makes sense and defines a probability measure $M_N(\ccdot;q,t;a,b,c,d)$ on $\Om_N$ with strictly positive weights.
\end{proposition}

\begin{proof}
As in the case $N=1$, the first two constraints in \eqref{eq5.O} ensure the strict positivity of the unnormalized weights
$$
\wt M_N(X;q,t;\abcd):=C_N(|X^+|;\zeta_+,\zeta_-)\,V_{q,t}(X)\prod_{x\in X} W(x; q; a,b,c,d), \qquad X\in\Om_N.
$$
Next, by virtue of the estimate in Proposition \ref{prop5.B} (i), the third constraint guarantees that the unnormalized weights are summable over $\Om_N$. 
\end{proof}

\begin{remark}\label{rem5.B}
As in the case $N=1$, one can also introduce the complementary series of multidimensional hypergeometric measures. For instance, one can allow the parameters $(a,b)$ to take real values subject to some constraints depending on $q$, $t$, $\zeta_+$, $\zeta_-$, and $N$. However, in the context of Theorem \ref{thm6.A}, one cannot satisfy these constraints for all $N$, unless $\tau$ is rational, see Remark \ref{rem6.A}. 
\end{remark}

\section{Main result: large-$N$ limit transition}\label{sect6}

\subsection{Preliminaries}\label{sect6.1}
We partly repeat and extend the definitions of subsection \ref{sect1.1.1}.

\subsubsection{Stochastic links}
Given $X\in\wt\Om$, we denote by $[X]$ the smallest closed interval of $\R$ containing $X$. 

According to \cite[Theorem A]{Ols-MacdonaldOne}, for each $N=2,3,\dots$ there exists a stochastic matrix $\LaN$ of format\/ $\Om_N\times \Om_{N-1}$,  with the following properties (below we denote by $\LaN(X,Y;q,t)$ the $(X,Y)$-entry of $\LaN$):

(1) the entry $\LaN(X,Y;q,t)$ vanishes whenever $Y$ is not contained in $[X]$;

(2)  for any $X\in\Om_N$ and any partition $\nu$ with $\ell(\nu)\le N-1$,
\begin{equation}\label{eq6.A1}
\sum_{Y \in\Om_{N-1}}\LaN(X,Y;q,t) \frac{P_{\nu\mid N-1}(Y;q, t)}{(t^{N-1};q, t)_\nu}=\frac{P_{\nu\mid N}(X;q,t)}{(t^N;q, t)_\nu}.
\end{equation}

An explicit expression for the entries $\LaN(X,Y;q,t)$ is given in \cite{Ols-MacdonaldOne} but we do not need it. We will use only the properties (1) and (2), which determine $\LaN$ uniquely.

\subsubsection{Coherent systems of measures}\label{sect6.1.2}

Recall that the symbol $\P(\ccdot)$  denotes the space of probability Borel measures on a given Borel (=measurable) space. We say that an infinite sequence $\{M_N\in\P(\Om_N): N\in\Z_{\ge1}\}$ forms a \emph{coherent system} with respect to the stochastic links $\LaN$ if  
$$
\sum_{X\in\Om_N}M_N(X)\LaN(X,Y;q,t)=M_{N-1}(Y) \qquad \text{for all $N\ge2$ and $Y\in\Om_{N-1}$}.
$$

By \cite[Theorem B]{Ols-MacdonaldOne}, there exists a one-to-one correspondence between coherent systems $\{M_N\}$ and probability measures $M_\infty\in\P(\Om_\infty)$. Under this correspondence one has
$$
M_\infty\La^\infty_N=M_N \qquad \text{for all $N\in\Z_{\ge1}$},
$$
where $\La^\infty_N: \Om_\infty\to\P(\Om_N)$ is a certain stochastic link (Markov kernel) between $\Om_\infty$ and $\Om_N$. In more detail, this relation reads as
$$
\int_{X\in\Om_N}M_\infty(dX)\La^\infty_N(X,Y;q,t)=M_{N-1}(Y) \qquad \text{for all $N\in\Z_{\ge1}$ and $Y\in\Om_N$}.
$$
The links $\La^\infty_N$ are characterized by the properties:

(1) $\La^\infty_N(X,Y;q,t)$ vanishes whenever $Y$ is not contained in $[X]$;

(2)  for any $X\in\Om_\infty$ and any partition $\nu\in\Y(N)$, one has 
\begin{equation}\label{eq6.B1}
\sum_{Y \in\Om_N}\La^\infty_N(X,Y;q,t) \frac{P_{\nu\mid N}(Y;q, t)}{(t^{N};q, t)_\nu}=P_\nu(X;q,t),
\end{equation}
where $P_\nu(X;q,t)$ is the Macdonald symmetric function with index $\nu$, evaluated at $X$.

We call $M_\infty$ the \emph{boundary measure} of $\{M_N\}$. 

\subsubsection{Large-$N$ limit transition}\label{sect6.1.3}

Let $\{M_N\in\P(\Om_N): N\in\Z_{\ge1}\}$ be a coherent system and $M_\infty\in\P(\Om_\infty)$ be the corresponding boundary measure. Since both $\Om_\infty$ and all $\Om_N$ are subsets of $\wt\Om$, we may put all the measures in the ambient space $\wt\Om$. Recall that it has a topology and $\P(\wt\Om)$ is equipped with the corresponding weak topology.

Now Theorem C from \cite{Ols-MacdonaldOne} tells that $M_N\to M_\infty$ as $N\to\infty$, in the weak topology of $\P(\wt\Om)$. This is an abstract result guaranteeing the existence of a large-$N$ limit, which  we will apply to concrete examples.

\subsection{The measures $\MM^{q,t;\alde}_N$}\label{sect6.2}

As before, we fix the parameters $q\in(0,1)$ and $t=q^\tau$, where $\tau>0$, and also  $\zeta_+>0$ and $\zeta_-<0$. 

So far $N$ was fixed, but now it will vary, and the quadruple $(\abcd)$  will vary together with it. Namely, we set
\begin{equation}\label{eq6.A}
(\abcd):=(\al,\be,\ga t^{1-N}, \de t^{1-N}),
\end{equation}
where $(\alde)$ is a new \emph{fixed} quadruple of parameters satisfying one of the following two collections of constraints:

\emph{Degenerate series}: 
\begin{equation*}
\al\in \zeta_+^{-1} q^\Z, \quad \be\in\zeta_-^{-1} q^\Z, \quad \ga=\bar\de\in\C\setminus\R.
\end{equation*}

\emph{Principal series}:
\begin{equation}\label{eq6.princ}
\al=\bar\be\in\C\setminus\R, \quad \ga=\bar\de\in\C\setminus\R, \quad \al\be<\ga\de q.
\end{equation}

\begin{definition}
For each $N\in\Z_{\ge1}$ we define a probability measure on $\Om_N$ by 
\begin{equation}\label{eq6.N-1}
\MM^{q,t;\alde}_N(X):=M_N(X;q,t;\al,\be,\ga t^{1-N}, \de t^{1-N}), \quad X\in\Om_N.
\end{equation}
In more detail (cf. \eqref{eq5.B}), 
\begin{equation}
\MM^{q,t;\alde}_N(X):=(Z^{q,t;\alde})^{-1} \wt \MM^{q,t;\alde}_N(X),
\end{equation}
where
\begin{equation}
\wt \MM^{q,t;\alde}_N(X):=C_N(|X^+|;\zeta_+,\zeta_-)\,V_{q,t}(X)\prod_{x\in X} W(x; q; \al,\be, \ga t^{1-N}, \de t^{1-N})
\end{equation}
and 
\begin{equation}\label{eq6.Z_N}
Z^{q,t;\alde}_N(q,t;a,b,c,d):=\sum_{X\in\Om_N}\wt \MM^{q,t;\alde}_N(X).
\end{equation}
\end{definition}

The definition makes sense, because for both series, degenerate  and principal, the varying parameters $(\abcd)$ given by \eqref{eq6.A} satisfy the necessary constraints for all $N$. 

\begin{remark}\label{rem6.A}
For rational $\tau$, one can extend the range of quadruples $(\alde)$ by allowing some of them to be real, so that the corresponding measures $\MM^{q,t;\alde}_N$ be in the complementary series for each $N=1,2,3,\dots$ (but for irrational $\tau$ this is impossible). 
\end{remark}

\subsection{The coherency relation}\label{sect6.3}

Below we use the results from \cite{Ols-MacdonaldOne} summarized in subsection \ref{sect6.1}. We only abbreviate the notation  $\LaN(X,Y;q,t)$ to $\LaN(X.Y)$.  

\begin{theorem}\label{thm6.A}
Let $(\alde)$ be a fixed quadruple of parameters from the degenerate  or principal series, and let $N=1,2,\dots$\,. The measures $\MM^{q,t;\alde}_N$ form a coherent system.
\end{theorem}

Recall that by definition this means that the following coherency relations hold
\begin{equation}\label{eq6.B}
\sum_{X\in\Om_N}\MM^{q,t;\alde}_N(X)\LaN(X,Y)=\MM^{q,t;\alde}_{N-1}(Y), \qquad N\ge2, \quad Y\in\Om_{N-1}.
\end{equation}

Note that for rational $\tau$, the theorem holds true for sequences $\{\MM^{q,t;\alde}_N\}$ from the complementary series (see Remark \ref{rem6.A}); the proof remains the same. 

\begin{proof} 
(i) We examine first the degenerate series case. It is convenient to write \eqref{eq6.B} in the condensed form
\begin{equation}\label{eq6.C}
\MM^{q,t;\alde}_N\LaN=\MM^{q,t;\alde}_{N-1}, \quad N\ge2,
\end{equation}
where we treat $\MM^{q,t;\alde}_N$ and $\MM^{q,t;\alde}_{N-1}$ as row vectors with coordinates indexed by $\Om_N$ and $\Om_{N-1}$, respectively. 
We also use the shorthand notation for the $N$-variate big $q$-Jacobi polynomials (see \eqref{eq3.C})
\begin{equation}\label{eq6.PhiN}
\Phi_{\la\mid N}:=\varphi_{\la\mid N}(\ccdot;q,t; \al,\be,\ga q^{\tau(1-N)}, \de q^{\tau(1-N)}).
\end{equation}

Because $\LaN$ is a stochastic matrix, $\MM^{q,t;\alde}_N\LaN$ \footnote{There is a typo in this expression in the published version.} is a probability measure. Next, we claim that it is concentrated on the subset $\Om_{N-1}[\be^{-1}q,\al^{-1}q]\subset\Om_N$, meaning that $\MM^{q,t;\alde}_N\LaN(Y)$ vanishes unless $Y$ is contained in $[\be^{-1}q,\al^{-1}q]$. 

Indeed, we know that $\MM^{q,t;\alde}_N$ \footnote{There is a typo in this expression in the published version.} is concentrated on $\Om_N[\be^{-1}q,\al^{-1}q]$. On the other hand, as pointed out above, $\LaN(X,Y)$ vanishes unless $Y\subset [X]$, which proves the claim. 

On the other hand, we know that $\MM^{q,t;\alde}_{N-1}$ is an orthogonality measure for the polynomials $\Phi_{\la\mid N-1}$ with $\la\in\Y(N-1)$ (Theorem \ref{thm5.A}); moreover, it is a unique probability measure on $\Om_{N-1}[\be^{-1}q,\al^{-1}q]$ which is orthogonal to all polynomials $\Phi_{\la\mid N-1}$ with nonzero index  (Proposition \ref{prop5.A}). 

Therefore, to prove the equality \eqref{eq6.C} it suffices to show that $\MM^{q,t;\alde}_N\LaN$ has the same orthogonality property, which can be written as 
\begin{equation}\label{eq6.D}
\langle \MM^{q,t;\alde}_N\LaN, \; \Phi_{\la\mid N-1}\rangle=0, \qquad \forall \la\in\Y(N-1), \quad \la\ne0,
\end{equation}
where the angular brackets denote the canonical pairing between measures and functions. 

It is convenient to regard $\Phi_{\la\mid N-1}$ as a column vector with the coordinates indexed by the set $\Om_{N-1}$, and then treat $\langle\ccdot,\ccdot\rangle$ as the canonical pairing between row and column vectors. Then we can rewrite \eqref{eq6.D} as 
\begin{equation}\label{eq6.E}
\langle \MM^{q,t;\alde}_N, \; \LaN \Phi_{\la\mid N-1}\rangle=0, \qquad \forall \la\in\Y(N-1), \quad \la\ne0,
\end{equation}
We will show that $\LaN \Phi_{\la\mid N-1}$ is proportional to $\Phi_{\la\mid N}$, and then \eqref{eq6.E} will follow, due to the orthogonality property of $\MM^{q,t;\alde}_N$.

More precisely, we are going to show that the following equality holds
\begin{equation}
\LaN \, \frac{\Phi_{\la\mid N-1}}{(t^{N-1};q,t)_\la}=\frac{\Phi_{\la\mid N}}{(t^N;q)_\la}, \qquad \la\in\Y(N-1).
\end{equation}
Indeed, Corollary \ref{cor3.A} implies (see Remark \ref{rem3.A}) that in the expansions
$$
\frac{\Phi_{\la\mid N}}{(t^N;q,t)_\la}=\sum_{\nu\subseteq\la} (\ccdot)\frac{P_{\nu\mid N}(\ccdot;q,t)}{(t^N;q,t)_\nu}
$$
and
$$
\frac{\Phi_{\la\mid N-1}}{(t^{N-1};q,t)_\la}=\sum_{\nu\subseteq\la} (\ccdot)\frac{P_{\nu\mid N-1}(\ccdot;q,t)}{(t^{N-1};q,t)_\nu}
$$
the coefficients are the same stable quantities $\pi(\la,\nu;q,t;\alde)$ that do not depend on the number of variables. 

On the other hand, we have
$$
\LaN \, \frac{P_{\nu\mid N-1}(\ccdot;q,t)}{(t^{N-1};q,t)_\nu}=\frac{P_{\nu\mid N}(\ccdot;q,t)}{(t^{N};q,t)_\nu}, \qquad \nu\in\Y(N-1),
$$
see  \eqref{eq6.A1}.

Combining these two facts we obtain the desired equality \eqref{eq6.E}. This completes the proof for the degenerate series.

(ii) Now we turn to the principal series. The desired coherency relation can be written in the form 
\begin{equation}\label{eq6.F}
Z^{q,t;\alde}_{N-1}\sum_{X\in\Om_N} \wt \MM^{q,t;\alde}_N(X)\LaN(X,Y)=Z^{q,t;\alde}_N \wt \MM^{q,t;\alde}_{N-1}(Y), \quad Y\in\Om_{N-1}.
\end{equation}

We are going to derive \eqref{eq6.F} from item (i) via analytic continuation. To do this we fix $(\ga,\de)$ keeping the  condition $\ga=\bar\de\in\C\setminus\R$, but relax the constraints on $(\al,\be)$. Namely, we will assume that $(\al,\be)$ ranges over the bidisc
$$
D:=\{(\al,\be)\in\C^2: |\al\be|<\ga\de q\}.
$$
Our aim is to show that both sides of \eqref{eq6.F} are well defined and are holomorphic functions in $D$. Once this is done, we can conclude that \eqref{eq6.F} holds as an equality of holomorphic functions in $D$. Indeed, we know from item (i) that this holds true on the infinite subset
$$
\{(\al,\be)\in\ D: \al\in\zeta_+^{-1}q^\Z,  \;\be\in\zeta_-^{-1}q^\Z\}.
$$
which is a uniqueness set for holomorphic functions in $D$.

The final steps are the following:

\smallskip

(1) For $X$ and $Y$ fixed, the quantities $\wt \MM^{q,t;\alde}_N(X)$ and $\wt \MM^{q,t;\alde}_{N-1}(Y)$ are well defined and are holomorphic functions of the variables $(\al,\be)$ ranging over $D$: this is clear from the very definition. 

(2) Note that $0\le \LaN(X,Y)\le1$. 

(3) The series 
$$
\sum_{X\in\Om_N} \wt \MM^{q,t;\alde}_N(X) \quad \text{and} \quad \sum_{Y\in\Om_{N-1}} \wt \MM^{q,t;\alde}_{N-1}(Y)
$$
converge absolutely and uniformly on $(\al,\be)$ ranging over an arbitrary compact subset of the domain $D$ --- here we use claim (ii) of Proposition \ref{prop5.B}. 

\smallskip
 
From (1) and (3) we see that the normalization constants are holomorphic functions with respect to $(\al,\be)\in D$. Next, from (1) and (2) we see that the same holds for the two series in \eqref{eq6.F}. We conclude that the both sides of \eqref{eq6.F} are holomorphic functions in $(\al,\be)\in D$. This completes the proof.  
\end{proof}

\begin{remark}
The above argument shows that the value of the normalization constant for the principal series measures can be obtained from the one for the degenerate series measures (Remark \ref{rem5.A}) via analytic continuation. 
\end{remark}

\subsection{Main result}\label{sect6.4}
The next result is a direct consequence of Theorem \ref{thm6.A} and \cite[Theorem C]{Ols-MacdonaldOne} (see subsection \ref{sect6.1.3} above).

\begin{theorem}\label{thm6.B}
Let, as above, $(\alde)$ be a fixed quadruple of parameters from the degenerate  or principal series.  As $N\to\infty$, the measures $\MM^{q,t;\alde}_N$ converge to a probability measure $\MM^{q,t;\alde}_\infty$ on $\Om$, in the weak topology of measures on the ambient space $\wt\Om$.
\end{theorem}

\section{Concluding remarks}\label{sect7}

\subsection{Big $q$-Jacobi symmetric functions}

The main theorem together with the computation of section \ref{sect3} allow us to extend the results of the author's paper \cite{Ols-FAA-2017} dealing with the $t=q$ case to the case of general parameters $q,t\in(0,1)$. We do not give proofs as they are similar to the arguments of that paper.  

Let $\Sym$ denote the algebra of symmetric functions. The Macdonald symmetric functions $P_\nu=P_\nu(\ccdot;q,t)$, indexed by partitions $\nu\in\Y$, form a homogeneous basis of $\Sym$. Fix a quadruple $(\alde)$ from the degenerate series. 

The \emph{big $q$-Jacobi symmetric functions} are defined by 
\begin{equation}
\Phi_\la^{q,t;\alde}(\ccdot):=\sum_{\nu\subseteq\la}
\pi(\la,\nu;q,t;\alde) P_\nu(\ccdot;q,t), \qquad \la\in\Y,
\end{equation}
where the coefficients $\pi(\la,\nu;q,t;\alde)$ are the same as in \eqref{eq3.stability}.  For $t=q$ this definition coincides with that of \cite{Ols-FAA-2017}.

These symmetric functions are large-$N$-limits of big $q$-Jacobi polynomials,
$$
\Phi_\la^{q,t;\alde}(\ccdot)=\lim_{N\to\infty}\varphi_{\la\mid N}(\ccdot;q,t;\al,\be, \ga t^{1-N},\de t^{1-N}), 
$$
where the limit is understood as in \cite[Definition 2.4]{Ols-FAA-2017}. 

By virtue of Corollary \ref{cor3.B},
$$
\Phi_\la^{q,t;\alde}(\ccdot)=P_\la(\ccdot;q,t)+\text{lower degree terms},
$$
so that the elements $\Phi_\la^{q,t;\alde}$ form an inhomogeneous  basis in $\Sym$.

Observe that elements of $\Sym$ can be realized as continuous functions on the space $\wt\Om$, and they are bounded on the set
$$
\{X\in\Om_\infty: X\subset[\be^{-1}q, \al^{-1}q]\},
$$
on which the limit measure $\MM^{q,t;\alde}_\infty$ is concentrated. 

\begin{theorem}\label{thm7.A}
The big $q$-Jacobi symmetric functions $\Phi^{q,t;\alde}_\nu$ form an orthogonal basis of the Hilbert space $L^2(\Om_\infty, \MM^{q,t;\alde}_\infty)$.
\end{theorem}

\subsection{Degeneration to discrete beta-ensembles}\label{sect7.2}
As was shown by Borodin and the author \cite{Ols-2003a}, \cite{BO-2005}, the problem of harmonic analysis on the infinite-dimensional unitary group $U(\infty)$ leads to a concrete model of discrete beta-ensembles with $\tau=1$. A more general model with arbitrary $\tau>0$ was studied by the author in \cite{Ols-2003b}. Following the notation of \eqref{eq1.Vdiscr1}, \eqref{eq1.Vdiscr2}, \eqref{eq1.Vdiscr} it can be defined as follows
\begin{multline}\label{eq7.A}
\PP_N^{\tau;\zw}(\nu):=\frac1{Z_N(\tau;\zw)}\prod_{1\le i<j\le N}\frac{\Ga(n_i-n_j+1)\Ga(n_i-n_j+\tau)}{\Ga(n_i-n_j)\Ga(n_i-n_j+1-\tau)}\\
\times \prod_{i=1}^N\frac{\Ga(-z+(1-N)\tau+n_i)\Ga(-z'+(1-N)\tau+n_i)}{\Ga(w+1+n_i)\Ga(w'+1+n_i)},
\end{multline} 
where 
$$
\nu=(\nu_1\ge\dots\ge\nu_N)\in\Z^N, \quad n_i:=\nu_i+(N-i)\tau,
$$
$(\zw)$ is a quadruple of parameters, and $Z_N(\tau;\zw)$ is the normalization factor. 

There are again the principal, degenerate, and (in the case of rational $\tau$) complementary series of parameters for which $\PP^{\tau;\zw}_N$ is a probability measure. For instance, the principal series is defined by the conditions
\begin{equation}\label{eq7.princ}
z=\overline{z'}\in\C\setminus\R, \quad w=\overline{w'}, \quad \Re(z+w)>-\tfrac12. 
\end{equation}
Note that there exists an explicit multiplicative expression for the normalization factor.   

The measures \eqref{eq7.A} satisfy the coherency relations
\begin{equation}\label{eq7.B}
\PP_N^{\tau;\zw} L^N_{N-1}=\PP_{N-1}^{\tau;\zw}, \qquad N\ge2,
\end{equation}
with the stochastic links \cite[(4.8)]{Ols-2003b}
\begin{multline}\label{eq7.C}
L^N_{N-1}(\nu,\mu):=
\prod_{1\le i<j\le N}\frac{\Ga(\mu_i-\nu_j+(j-i)\tau)}{\Ga(\mu_i-\nu_j+(j-i-1)\tau+1)}\, \cdot
\prod_{1\le i\le j\le N-1}\frac{\Ga(\nu_i-\mu_j+(j-i+1)\tau)}{\Ga(\nu_i-\mu_j+(j-i)+1)}\\
\times \prod_{1\le i<j\le N}\frac{\Ga(\nu_i-\nu_j+(j-i-1)\tau+1)}{\Ga(\nu_i-\nu_j+(j-i+1)\tau)} \,\cdot
\prod_{1\le i<j\le N-1}(\mu_i-\mu_j+(j-i)\tau) \,\cdot \frac{\Ga(N\tau)}{(\Ga(\tau))^N}\,\cdot\one_{\mu\prec\nu},
\end{multline}
where $\mu=(\mu_1\ge\dots\ge\mu_{N-1})\in\Z^{N-1}$ and the symbol $\one_{\mu\prec\nu}$ equals $1$ or $0$ depending on whether $\mu$ and $\nu$ interlace (i.e. $\la_i\ge\mu_i\ge\la_{i+1}$ for $i=1,\dots,N-1$)  or not. 

In the case $\tau=1$, the right-hand side of  \eqref{eq7.C} substantially simplifies and reduces to
$$
(N-1)!\,\dfrac{\prod\limits_{1\le i<j\le N-1}(\mu_i-\mu_j+j-i)}{\prod\limits_{1\le i<j\le N}(\nu_i-\nu_j+j-i)}\, \cdot \one_{\mu\prec\nu}.
$$
In this case, the coherency relation \eqref{eq7.B} was first obtained in \cite{Ols-2003a} from a representation-theoretic construction. 

In the case of arbitrary $\tau>0$, the coherency relation was proved in \cite{Ols-2003b} by combinatorial methods, in three different ways. Some ideas of \cite{Ols-2003b} are used in the present paper.  

Note that the fact that \eqref{eq7.C} is a stochastic matrix is also nontrivial.  

\begin{proposition}\label{prop7.A}
Assume for definiteness that $(\zw)$ satisfies the constraints \eqref{eq7.princ}. Fix $\tau>0$ and set, for $q\in(0,1)$,  
\begin{equation}\label{eq7.abcd}
t=q^\tau, \quad \al=q^{w+1}, \quad \be=q^{w'+1}, \quad \ga=q^{-z}, \quad \de=q^{-z'}.
\end{equation}

For each $N$ and each $N$-tuple of integers $\nu=(\nu_1\ge\dots\ge\nu_N)$  the following limit relation holds
\begin{equation}\label{eq7.DegOne}
\lim_{q\to1}\MM^{q,t;\alde}_N(q^{\nu_N}, q^{\nu_{N-1}+\tau}, \dots,q^{\nu_1+(N-1)\tau})=\PP^{\tau; \, z,z',w,w'}_N(\nu).
\end{equation}
\end{proposition}

Note that the principal series measures $\MM^{q,t;\alde}_N$ with parameters \eqref{eq7.abcd} are well defined, at least for $q$ close enough to $1$: indeed, from \eqref{eq7.princ} it is seen that the quadruple $(\alde)$ defined by  \eqref{eq7.abcd} satisfies the necessary constraints \eqref{eq6.princ}. 

Because both $\PP^{\tau; \, z,z',w,w'}_N$ and $\MM^{q,t;\alde}_N$ are probability measures, the limit relation \eqref{eq7.DegOne} implies that,  for an arbitrarily small fixed interval $(1-\eps,1+\eps)$ around the point $1$, 
$$
\lim_{q\to1}\sum_{\{X\in\Om_N: \; X\subset (1-\eps,1+\eps)\}}\MM^{q,t;\alde}(X)=1.
$$
In particular, the configurations containing particles on the left of zero are asymptotically negligible.  

Thus, $\PP^{\tau;s}_N$ can be obtained from $\MM^{q,t;\alde}$ via a scaling limit transition (we have to focus on a small vicinity of $1$ and zoom in with a large scale factor of order $\log(q^{-1})$). 

\subsection{Degeneration to continuous beta-ensembles}\label{sect7.3}

Fix $s\in\C$ with $\Re(s)>\frac12$ and $\tau>0$. We define the $N$-dimensional \emph{s-measure} (or \emph{generalized Cauchy distribution})  as the following probability measure on the domain $\mathcal D_N:=\{u\in\R^N: u_1>\dots>u_N\}$ in $\R^N$:
\begin{equation}\label{eq7.D}
P^{\tau;s}_N(du)=\frac{1}{Z_N(s,\tau)} \prod_{k=1}^N\frac{du_k}{(1+iu_k)^{s+(N-1)\tau}(1-iu_k)^{\bar s+(N-1)\tau}}\cdot\prod_{1\le k<\ell\le N}(u_k-u_\ell)^{2\tau}, \quad u\in\mathcal D_N,
\end{equation}
where $Z_N(s,\tau)$ is the normalization constant and $du=du_1\dots du_N$. The s-measures were examined in Borodin--Olshanski \cite{BO-2001} and Witte--Forrester \cite{WF} in the case $\tau=1$, and in Neretin \cite{Ner-2003} and Assiotis--Najnudel \cite{AN} for arbitrary $\tau>0$. 

The s-measures satisfy the coherency relations of the form
\begin{equation}\label{eq7.E}
\int_{u\in\mathcal D_N} P^{\tau;s}_N(du) \mathcal L(u,dv)=P^{\tau;s}_{N-1}(dv), \qquad N\ge2,
\end{equation}
with the Markov kernels $\mathcal L^N_{N-1}(u,dv)$ on $\mathcal D_N\times\mathcal D_{N-1}$ coming from the integrand of the  \emph{Dixon--Anderson integral} (Forrester--Warnaar \cite[section 2.1]{ForrWar}):
\begin{multline}\label{eq7.F}
\mathcal L(u,dv):=\frac{\Ga(N\tau)}{(\Ga(\tau))^N}\prod_{1\le i<j\le N}(u_i-u_j)^{1-2\tau}\,\cdot \prod_{1\le i<j\le N-1}(v_i-v_j)\\
\times \prod_{i=1}^{N-1}\prod_{j=1}^N |v_i-u_j|^{\tau-1}\,\cdot\one_{v\prec u}\,\cdot dv,
\end{multline}
where 
$$
v=(v_1>\dots>v_{N-1})\in\mathcal D_{N-1}, \qquad dv:=dv_1\dots dv_{N-1},
$$
and $\one_{v\prec u}$ expresses the interlacement condition $u_i>v_i>u_{i+1}$. 

The coherency relation \eqref{eq7.E} follows from a more general integral evaluated in Lemma 2.2 of Neretin's paper \cite{Ner-2003}. 

Now we fix the parameters $s$ and $\tau$, set
$$
\al=-iq^{s}, \quad \be=iq^{\bar s}, \quad \ga=-i, \quad \de=i,
$$
and consider the corresponding measures $\MM^{q,t;\alde}_N$, $N=1,2,\dots$\,. Again, they are well defined, at least for $q$ close to $1$: indeed, then the constraints \eqref{eq6.princ} are satisfied due to the condition $\Re s>\frac12$. 

\begin{proposition}\label{prop7.B}
For each $N$, as $q\to1$, the measures\/ $\MM^{q,t;\alde}_N$ weakly converge to the s-measure $P^{\tau;s}_N$. 
\end{proposition}

Note that in this limit transition (in contrast to that of Proposition \ref{prop7.A}) there is no scaling.

The proof is based on the limit relation (Gasper--Rahman \cite[ch. 1, (3.19)]{GR})
$$
\lim_{q\to1}\frac{(vq^A;q)_\infty}{(vq^B;q)_\infty}={(1-v)^{B-A}}.
$$

Finally, note that the Markov kernels $\mathcal L^N_{N-1}$ are readily obtained from the stochastic matrices $L^N_{N-1}$ (formula \eqref{eq7.C}) by a scaling limit transition which turns the lattice $\Z$ into the real line $\R$. In this way one can also degenerate the measures $\PP^{\tau; \, z,z',w,w'}_N$ into the s-measures  $P^{\tau;s}_N$ and obtain the coherency relation \eqref{eq7.E} from the coherency relation \eqref{eq7.B}.

\bigskip

\noindent \textbf{Funding}

\medskip

\noindent This work was supported by the Russian Science Foundation, project 20-41-09009.

\bigskip

\noindent\textbf{Acknowledgments}
\medskip

\noindent I am grateful to Cesar Cuenca and the anonymous referee for valuable comments.

\end{document}